\pdfoutput=1

\documentclass{article}
\usepackage{array}
\usepackage[utf8]{inputenc}
\usepackage[table]{xcolor}
\usepackage{amssymb}
\usepackage{authblk}
\usepackage{setspace}
\usepackage[margin=1.25in]{geometry}
\usepackage{graphicx}
\graphicspath{ {./figures/} }
\usepackage{subcaption}
\usepackage{amsmath}
\usepackage{lineno}
\usepackage{caption}
\usepackage{hyperref}
\usepackage{tabularx}
\usepackage{booktabs}
\usepackage{rotating}
\usepackage[most]{tcolorbox}
\usepackage{tikz}
\usepackage{amsmath}
\usepackage{amsthm}
\usepackage{algorithm}
\usepackage{algpseudocode}
\usepackage{soul}
\usepackage{bm}
\usepackage{multirow}
\usepackage{enumitem}
\usetikzlibrary{arrows.meta,positioning,calc}

\algnewcommand\algorithmicinput{\textbf{Input:}}
\algnewcommand\algorithmicoutput{\textbf{Output:}}
\algnewcommand\Input{\item[\algorithmicinput]}
\algnewcommand\Output{\item[\algorithmicoutput]}

\theoremstyle{definition}
\newtheorem{definition}{Definition}[section]

\theoremstyle{plain}
\newtheorem{proposition}[definition]{Proposition}

\theoremstyle{remark}
\newtheorem{remark}[definition]{Remark}

\newcommand{\Agents}{\mathcal{A}}          
\newcommand{\Claims}{\mathcal{C}}          
\newcommand{\Evid}{\mathcal{E}}            
\newcommand{\Args}{\mathcal{G}}            
\newcommand{\Mem}{\mathcal{M}}

\newcommand{\Trace}{\mathcal{T}}           
\newcommand{\Just}{\mathcal{J}}            
         
\newcommand{\Media}{\mathcal{V}}           
\newcommand{\Labels}{\mathcal{Y}}          
\newcommand{\Objs}{\mathcal{Z}}            
\newcommand{\Urls}{\mathcal{W}}

\newcommand{\Learn}{\mathcal{D}_{\mathrm{learn}}}
 
\newcommand{\Ind}{\operatorname{Ind}}      
\newcommand{\Ext}{\operatorname{Ext}}      
\newcommand{\IDs}{\operatorname{IDs}}
\newcommand{\clip}{\operatorname{clip}}
\newcommand{\score}{\psi}                  
\newcommand{\infl}{g}                      
\newcommand{\cand}{\kappa}                 
\newcommand{\unc}{\omega}                  
\newcommand{\ctx}{\zeta}                   
\newcommand{\cfloor}[1]{\underline{\gamma}\_{#1}} 
\newcommand{\prov}{\mathrm{prov}}          
\newcommand{\spec}{\ell}                   
\newcommand{\dvs}{d}                       
\newcommand{\cratio}{\operatorname{cr}}        
\newcommand{\stC}{\mathsf{C}}\newcommand{\stR}{\mathsf{R}}
\newcommand{\stV}{\mathsf{V}}\newcommand{\stA}{\mathsf{A}}
\newcommand{\stQ}{\mathsf{Q}}\newcommand{\stF}{\mathsf{F}}
\newcommand{\stO}{\mathsf{O}}
\newcommand{\best}[1]{\textbf{#1}}
\newcommand{\secondbest}[1]{\underline{#1}}
\newcommand{\NA}{--}

\providecommand{\clip}{\operatorname{clip}}
\providecommand{\IDs}{\operatorname{IDs}}

\definecolor{SEMVNavy}{HTML}{183B56}
\definecolor{SEMVTeal}{HTML}{17807A}
\definecolor{SEMVGold}{HTML}{D89B32}
\definecolor{DraftRed}{HTML}{A33A31}
\definecolor{PaleGold}{HTML}{FFF3CC}
\definecolor{PaleBlue}{HTML}{EAF2F8}
\definecolor{PaleGreen}{HTML}{E7F3EF}
\definecolor{PalePurple}{HTML}{4D61DC}
\definecolor{PalePurpleRed}{HTML}{8F689E}

\providecommand{\changeblock}[4]{
  \par\smallskip\noindent
  \begingroup
  \setlength{\fboxsep}{3pt}
  \colorbox{#1}{
    \parbox{\dimexpr\linewidth-2\fboxsep\relax}{
      \raggedright
      \textcolor{#2}{\scriptsize\sffamily\bfseries[#3]}\enspace
      #4\strut
    }
  }
  \endgroup
  \par\smallskip
}

\providecommand{\retrievalblock}[1]{
  \changeblock{PaleBlue}{SEMVNavy}{C}{#1}}

\providecommand{\contestblock}[1]{
  \changeblock{DraftRed!6}{DraftRed}{E}{#1}}

\providecommand{\memoryblock}[1]{
  \changeblock{PaleGreen}{SEMVTeal}{F$\rightarrow$C}{#1}}

\newcommand{\mediaframe}[2]{
\begin{minipage}[t]{0.28\linewidth}
    \centering
    \includegraphics[
        width=\linewidth,
        height=2.75cm,
        keepaspectratio
    ]{#1}\\[-1pt]
    {\scriptsize\sffamily $t=#2$}
\end{minipage}
}

\newcommand{\keyframegap}{
\begin{minipage}[t]{0.02\linewidth}
    \centering
    \raisebox{1.25cm}{\large\(\cdots\)}
\end{minipage}
}

\newtcolorbox{pairedreport}[2][]{
  enhanced,
  breakable,
  colback=#2!3,
  colframe=#2,
  coltitle=white,
  boxrule=0.7pt,
  arc=2pt,
  left=7pt,
  right=7pt,
  top=6pt,
  bottom=6pt,
  fonttitle=\bfseries,
  fontupper=\footnotesize,
  title={#1}
}

\newtcolorbox{humancontestbox}{
  enhanced,
  breakable,
  colback=DraftRed!6,
  colframe=DraftRed,
  boxrule=0.7pt,
  arc=2pt,
  left=6pt,
  right=6pt,
  top=5pt,
  bottom=5pt,
  fontupper=\small,
  title={Human Contestation with Arguments},
  coltitle=white,
  fonttitle=\bfseries,
  fontupper=\footnotesize,
}

\newcommand{\haction}[1]{
  \textcolor{DraftRed}{\scriptsize\sffamily\bfseries[#1]}\enspace
}

\tcbuselibrary{theorems,breakable,skins}

\newtcbtheorem[number within=section]
  {selfevodefinition}
  {Definition}
  {
    enhanced,
    breakable,
    colback=blue!3,
    colframe=blue!45!black,
    boxrule=0.3pt,
    arc=2pt,
    left=6pt, right=6pt, top=5pt, bottom=5pt,
    fonttitle=\small\bfseries,
    coltitle=black,
    colbacktitle=blue!10,
    separator sign none,
    description font={\par\normalfont\itshape\small},
  }
  {def}
\usepackage[style=nejm, 
citestyle=numeric-comp,
sorting=none]{biblatex}
\title{Self-Evolving Multimedia Verification through\\Memory Consolidation of Contestation Experiences}

\author[1*]{Truong Thanh Hung Nguyen}
\author[2]{Vo Thanh Khang Nguyen}
\author[1]{Hoang-Loc Cao}
\author[1]{\protect\\Phuc Ho}
\author[3]{Truong Thinh Nguyen}
\author[1]{Van Pham}
\author[1*]{Hung Cao}

\affil[1]{University of New Brunswick, Fredericton, New Brunswick, Canada}
\affil[2]{FPT Software, Quy Nhon, Vietnam}
\affil[3]{University of Science and Technology of Hanoi, Hanoi, Vietnam}

\affil[*]{Address correspondence to: \{hung.ntt,hung.cao\}@unb.ca}

\date{}

\begin{document}

\maketitle

\begin{abstract}
Multimedia verification requires not only accurate decisions but also traceable evidence, reliable human correction, and safe reuse of prior experience. Existing systems often lack explicit mechanisms for revising intermediate reasoning or preventing harmful knowledge transfer. We present SEMV (\textit{Self-Evolving Multimedia Verification}), a self-evolving multi-agent framework that treats provenance-bearing arguments as the interface between evidence, reasoning, human contestation, and memory. SEMV combines arena-based quantitative bipolar argumentation (A-QBAF), causal and scoped revision, and verification-gated memory consolidation with explicit conflict retention. On COSMOS benchmark, SEMV achieves 91.88\% accuracy versus 89.10\% for the strongest comparable baseline. Verified memory reduces negative transfer from 5.7\% to 0.2\%. On CTR benchmark, constructed from reviewer contestations, scoped causal revision corrects 96.7\% of initial errors while saving 52.8\% compute. MV2026 Grand Challenge dataset further supports evidence-grounded, temporally consistent reporting. These results show that SEMV can evolve through verified experience while keeping accumulated knowledge and subsequent decisions traceable, revisable, and contestable.

\end{abstract}

\section{Introduction}

Multimedia verification determines not only whether images or videos are authentic, but also whether their accompanying narratives correctly describe what happened, where and when it occurred, and who was involved. Authentic media can therefore be misleading when reused in a false context, producing \textit{out-of-context} (OOC) cheapfakes \cite{aneja2023cosmos,dang20262026}. COSMOS \cite{aneja2023cosmos} studies this problem through image-caption consistency, while the Multimedia Verification Grand Challenges require evidence-rich investigation of provenance, time, location, entities, motivation, and forensic authenticity \cite{dang20262026}.

Current systems remain limited in how their verification state can be inspected, corrected, and reused. Classifiers compress heterogeneous evidence into a final label, while tool-augmented language models may generate detailed reports without explicit provenance and conflict constraints. Human feedback is also commonly applied only after inference, making it difficult to identify whether an error originates from decomposition, retrieval, validation, reasoning, or aggregation. These limitations are especially consequential in open-world verification, where evidence availability changes and correct decisions may depend on incomplete or conflicting sources.

We introduce SEMV (\textit{Self-Evolving Multimedia Verification})\footnote{
Our implementation is publicly available at
\url{https://github.com/Analytics-Everywhere-Lab/SEMV}}, a multi-agent framework in which the \emph{argument} serves as the interface between perception, external evidence, reasoning, and human review. Evidence-grounded support and attack arguments are associated with explicit subclaims, provenance, verification status, and decomposable strengths, then resolved using an arena-based quantitative bipolar argumentation framework (A-QBAF) \cite{cao2026neuro,baroni2019fine}. The resulting argumentative state directly governs subclaim decisions, uncertainty, contestation, and subsequent memory updates rather than serving only as a post-hoc explanation.

SEMV further introduces \emph{causal contestability} and \emph{verified self-evolution}. Human feedback is routed to the earliest affected pipeline stage, after which only its downstream dependency cone is recomputed while independent results are preserved \cite{lyons2021conceptualising,alfrink2023contestable,leofante2024contestable,nguyen2026heart2mind}. Experiences generated through verification and contestation are staged rather than immediately reused, then checked for provenance, support, duplication, and conflict before consolidation. This verification-gated policy reduces the risk that agents indiscriminately follow experiences outside their applicable context \cite{xiong2026memory}.

Building on our previous multi-agent verification framework \cite{le2025multimedia} and arena-based reasoning approach \cite{nguyen2026contestable}, this paper makes the following contributions:

\begin{enumerate}
\item We formalize SEMV as a \textit{cooperative multi-agent verification framework} with structured case state, typed agent communication, provenance-constrained arguments, A-QBAF propagation, uncertainty estimation, and deterministic decision aggregation (Sec.~\ref{sec:problem}--\ref{sec:method}).

\item We introduce \textit{adaptive causal contestation}, which routes interventions to affected pipeline stages, performs dependency-aware localized recomputation, preserves traces, and quantifies intervention locality (Sec.~\ref{subsec:contestation}).

\item We develop \textit{verification-gated self-evolution} through staged experience verification, multi-case consolidation, conflict retention, frozen evaluation snapshots, and metrics that distinguish beneficial transfer from negative transfer (Sec.~\ref{subsec:self_evolution}).

\item We establish a comprehensive evaluation of predictive performance, statistical significance, report quality, calibration, contestability, memory-transfer safety, and computational cost across Gemma4-31B \cite{gemma42026}, Nemotron-3-Nano-Omni-30B \cite{deshmukh2026nemotron}, Qwen3.6-35B \cite{qwen36card}, and InternVL3.5-38B \cite{wang2025internvl3_5}, together with qualitative trace analysis (Sec.~\ref{sec:experiments}--\ref{sec:results}).
\end{enumerate}

\section{Related Work}
\label{sec:related}

\subsection{Multimedia Verification}

Multimedia verification asks whether media is authentic and whether it is presented in its original factual context. \textit{Cheapfakes} include low-cost manipulations such as cropping, splicing, speed changes, selective editing, and re-staging, as well as misleading use of unaltered media \cite{paris2019deepfakes}. OOC misuse is a prominent case in which authentic media is falsely associated with an event, location, time, person, or motivation \cite{aneja2023cosmos,abdelnabi2022open}. Verification has therefore expanded from image-text consistency classification toward evidence-grounded investigation.

Existing methods broadly follow three directions. Representation and feature-fusion approaches combine cross-modal semantics, caption generation, natural-language or visual entailment, and auxiliary linguistic signals \cite{la2022multimodal,la2022leverage,tran2022textual,alkaddour2022sentiment}. External-evidence methods retrieve and assess web evidence using cycle consistency, relevance-aware reranking, or retrieval--stance--explanation pipelines \cite{abdelnabi2022open,papadopoulos2025red,yao2023mocheg}. More recent LLM/VLM systems enrich contextual reasoning through prompted relation extraction, generated visual context, multimodal comparison, source and metadata retrieval, and dynamic tool-assisted investigation \cite{wu2023cheap,vo2024detecting,seo2024multi,nguyen2025robust,braun2025defame}.

Recent multimedia verification challenges further emphasize case-level investigation and structured reporting. Multimedia Verification Grand Challenges \cite{dang20262026} introduced source discovery, manipulation analysis, and detailed reports, motivating systems that combine forensic, semantic, temporal, geospatial, OSINT, retrieval, and multimodal reasoning \cite{phan2025fact,pham2025aegis,le2025multimedia}. MV2026 formalizes these requirements through a professional rubric covering provenance, evidence URLs, location, time, entities, motivation, forensic authenticity, and report quality \cite{dang20262026}. Recent systems accordingly explore multilingual evidence-based reasoning \cite{nguyen2026deepverify} and multi-agent verification \cite{muneer2026mosaiv}. Our preliminary A-QBAF framework instead represents retrieved evidence as provenance-bearing support and attack arguments in local quantitative bipolar argument graphs \cite{nguyen2026contestable}. However, it does not address \textit{selective propagation of corrections} or \textit{safe reuse of verified experience}, which motivate our proposed SEMV. Table~\ref{tab:mv_system_taxonomy} summarizes these distinctions.

\begin{sidewaystable}[p]
\centering
\caption{Capability comparison of evidence-grounded multimedia verification approaches. $\checkmark$ denotes an explicit mechanism, $\triangle$ denotes partial, implicit, or terminal-stage support, and -- denotes that the capability is absent or not documented. The comparison concerns reported system mechanisms, not capabilities that a general-purpose backbone might exhibit without explicit control.}
\label{tab:mv_system_taxonomy}

\setlength{\tabcolsep}{3pt}
\renewcommand{\arraystretch}{1.15}

\begin{tabularx}{\linewidth}{
    p{3.3cm}
    *{12}{>{\centering\arraybackslash}p{.7cm}}
    X}
\toprule
\textbf{Approach} &
\textbf{MM} & \textbf{OS} & \textbf{PV} & \textbf{EV} &
\textbf{FA} & \textbf{GT} & \textbf{CD} & \textbf{MA} &
\textbf{AG} & \textbf{HI} & \textbf{RG} & \textbf{VM} &
\textbf{Verification Output} \\
\midrule
La \textit{et al.}\ \cite{la2022multimodal,la2022leverage} &
$\checkmark$ & -- & -- & $\triangle$ & -- & -- & -- & -- & -- & -- & -- & -- &
Cross-modal visual- caption-based verdict \\

Tran \textit{et al.}\ \cite{tran2022textual} &
$\checkmark$ & $\checkmark$ & $\triangle$ & $\checkmark$ & -- & -- & -- & -- & -- & -- & -- & -- &
Unsupervised ensemble verdict \\

Abdelnabi et al.\ \cite{abdelnabi2022open} &
$\checkmark$ & $\checkmark$ & $\triangle$ & $\checkmark$ &
-- & -- & -- & -- & -- & -- & -- & -- &
Evidence-conditioned pair verdict \\

MOCHEG \cite{yao2023mocheg} &
$\checkmark$ & $\checkmark$ & $\checkmark$ & $\checkmark$ &
-- & -- & -- & -- & -- & -- & $\triangle$ & -- &
Verdict and generated explanation \\

RED-DOT \cite{papadopoulos2025red} &
$\checkmark$ & $\checkmark$ & $\triangle$ & $\checkmark$ &
-- & -- & -- & -- & -- & -- & -- & -- &
Evidence-reranked pair verdict \\

DEFAME \cite{braun2025defame} &
$\checkmark$ & $\checkmark$ & $\checkmark$ & $\checkmark$ &
$\triangle$ & $\triangle$ & $\triangle$ & -- & -- & -- &
$\checkmark$ & -- &
Structured multimodal fact-check \\

Phan \textit{et al.}\ \cite{phan2025fact} &
$\checkmark$ & $\checkmark$ & $\checkmark$ & $\checkmark$ &
$\checkmark$ & $\checkmark$ & -- & $\triangle$ & -- & -- &
$\checkmark$ & -- &
Expert report and public summary \\

\textsc{Aegis} \cite{pham2025aegis} &
$\checkmark$ & $\checkmark$ & $\checkmark$ & $\checkmark$ &
$\checkmark$ & $\checkmark$ & -- & -- & -- & -- &
$\checkmark$ & -- &
Structured OSINT verification report \\

Le \textit{et al.}\ \cite{le2025multimedia} &
$\checkmark$ & $\checkmark$ & $\checkmark$ & $\checkmark$ &
$\triangle$ & $\checkmark$ & $\triangle$ & $\checkmark$ &
-- & -- & $\checkmark$ & -- &
Six-stage verification report \\

DeepVerify \cite{nguyen2026deepverify} &
$\checkmark$ & $\checkmark$ & $\checkmark$ & $\checkmark$ &
$\checkmark$ & $\triangle$ & -- & -- & -- & -- &
$\checkmark$ & -- &
Evidence-based explainable report \\

MOSAIV \cite{muneer2026mosaiv} &
$\checkmark$ & $\checkmark$ & $\triangle$ & $\triangle$ &
$\checkmark$ & $\checkmark$ & $\triangle$ & $\checkmark$ &
-- & -- & $\checkmark$ & -- &
Agent-swarm verification report \\

A-QBAF \cite{nguyen2026contestable} &
$\checkmark$ & $\checkmark$ & $\checkmark$ & $\checkmark$ &
$\triangle$ & $\checkmark$ & $\checkmark$ & $\checkmark$ &
$\checkmark$ & $\checkmark$ & $\checkmark$ & -- &
Contestable argument-grounded report \\

\rowcolor{PalePurple!6} \textbf{SEMV (Ours)} &
$\checkmark$ & $\checkmark$ & $\checkmark$ & $\checkmark$ &
$\checkmark$ & $\checkmark$ & $\checkmark$ & $\checkmark$ &
$\checkmark$ & $\checkmark$ & $\checkmark$ & $\checkmark$ &
Contestable self-evolving case report \\

\bottomrule
\end{tabularx}

\vspace{4pt}

\begin{minipage}{\linewidth}
\small
\textbf{Legend:}
MM = multimodal/VLM reasoning;
OS = open-Web or OSINT retrieval;
PV = explicit source provenance;
EV = evidence relevance, sufficiency, or stance validation;
FA = media forensic authenticity;
GT = geospatial and temporal reasoning;
CD = explicit claim decomposition;
MA = multi-agent orchestration;
AG = formal support--attack argumentation;
HI = human intervention or contestation;
RG = structured report generation;
VM = verified cross-case self-evolving memory.
For HI, $\triangle$ denotes terminal review without localized causal
revision, whereas $\checkmark$ denotes intervention over intermediate
verification objects.
\end{minipage}

\end{sidewaystable}

\subsection{From Explainability to Contestability}
\textit{Explainability} exposes the basis of a decision, whereas \textit{contestability} enables users to act on that basis by identifying disputed components, justifying challenges, and obtaining revised decisions while preserving the original trace, thereby supporting interaction, redress, and accountability throughout the decision lifecycle. Prior research \cite{lyons2021conceptualising,alfrink2023contestable,nguyen2026heart2mind} frames contestable AI as a dynamic human-machine process in which explanations are examined, objections are evaluated, and successful challenges trigger revision.

Computational argumentation operationalizes this process through explicit and editable reasoning objects. Abstract frameworks encode attacks among arguments, bipolar frameworks add support relations, and quantitative variants propagate numerical argument strengths toward a conclusion under a \textit{gradual semantics} \cite{baroni2019fine,rago2016discontinuity,rago2023interactive,cao2026neuro}.
Unlike free-form rationales, argument graphs expose which propositions support or attack a claim, how strongly they contribute, and how an intervention changes the inference. Argumentative language models have applied this structure to claim verification \cite{freedman_argumentative_2025}, while a growing family of systems now develops quantitative bipolar frameworks over retrieved evidence \cite{pmlr-v284-zhu25a}, globally contestable decision-support systems \cite{argeval2026}, and methods for adjusting argument or relation strengths to produce desired outcomes \cite{yin2026contestability,kampik2026strength}. SEMV differs from this line in the \emph{granularity} of the intervention: prior work contests the argument graph, whereas SEMV routes a contestation to the earliest \emph{pipeline stage} that can repair it, which may be claim decomposition, retrieval, or validation rather than the graph itself.

\subsection{Self-Evolution through Memory Consolidation}

Whereas contestability repairs the current case, memory-based adaptation transfers experience to future decisions. Early methods such as Reflexion, Agent-Pro, and SAGE use feedback and self-reflection to store episodic lessons, revise behavioral policies, and retain useful information across interactions \cite{shinn2023reflexion,zhang2024agentpro,liang2025sage}. Later approaches move toward cross-episode abstraction: ExpeL extracts reusable insights and examples \cite{zhao2024expel}, MetaReflection converts failed trials into semantic instructions \cite{gupta2024metareflection}, and Contextual Experience Replay synthesizes trajectory summaries in a dynamic buffer \cite{liu2025contextual}.

Recent work further studies memory organization and consolidation. A-MEM builds dynamically linked records \cite{xu2026amem}, RecMem consolidates only after semantically similar experiences recur \cite{dai2026recmem}, and ICAL converts inefficient visual trajectories and human feedback into corrected examples \cite{sarch2024vlm}. Beyond memory itself, agent symbolic learning optimizes prompts, tools, and workflows through language-based analogues of backpropagation \cite{ou2025symbolic}.
These systems also expose risks that motivate our verification-gated design. Repeated LLM rewriting can cause semantic drift \cite{zhang2026useful}, agents may follow retrieved experiences despite poor applicability \cite{xiong2026memory}, and locally valid but non-transferable memories can induce negative transfer. Accordingly, memory evaluation should separate immediate correction from downstream transfer and measure both benefit and harm. Long-term memory should additionally preserve provenance, scope, conflicts, and reuse outcomes \cite{wu2025longmemeval}.

We define self-evolution as a closed loop in which contestation and outcome-grounded experience produce candidate knowledge that is verified, consolidated into external memory, and reused in future decisions, while model parameters and agent structure remain fixed (Fig.~\ref{fig:self_evolving_memory_consolidation}). Accordingly, we adopt the following operational definition:

\begin{selfevodefinition}
{Self-Evolution through Memory Consolidation with Contestation Experiences}{selfevo}
A \textit{self-evolving agent through memory consolidation with contestation experiences} persistently adapts its future decisions through a closed-loop, non-parametric process in which autonomous outcomes and human contestations produce provenance-preserving experience, candidate lessons are verified and consolidated across cases into external memory, and relevant knowledge is selectively reused and attributed in subsequent decisions without modifying the underlying model parameters.
\end{selfevodefinition}

\begin{figure}[t]
    \centering
    \includegraphics[width=.9\linewidth]{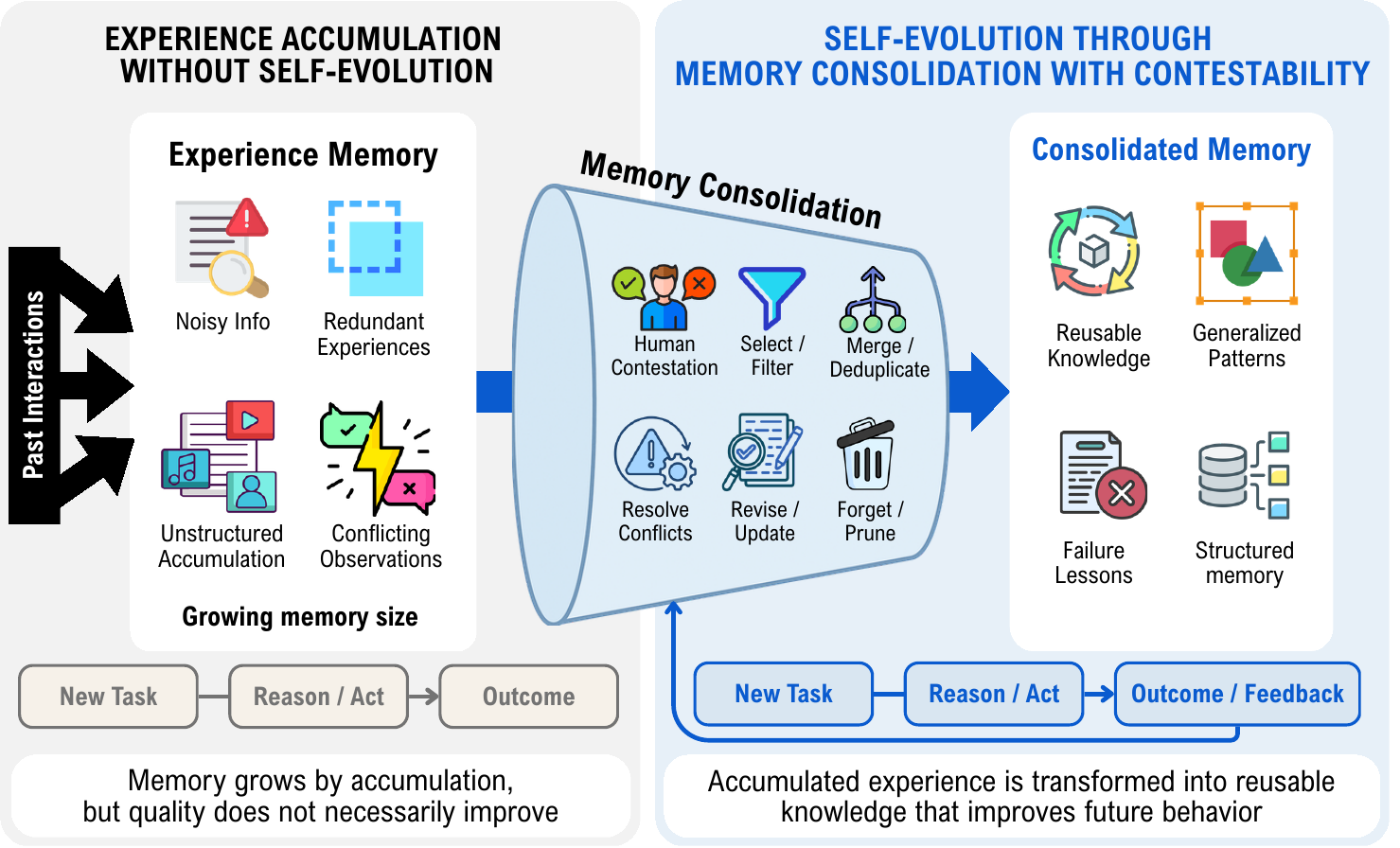}
    \caption{Conceptual design of \textit{Self-Evolution through Memory Consolidation with Contestation Experiences} framework. Without consolidation, past interactions accumulate as potentially noisy, redundant, unstructured, or conflicting experiences. With consolidation, human contestation, and memory operations such as filtering, deduplication, conflict resolution, revision, and pruning, accumulated experience is transformed into structured reusable knowledge, generalized patterns, and failure lessons for future tasks.}
\label{fig:self_evolving_memory_consolidation}
\end{figure}

\textit{Memory consolidation} goes beyond accumulation by filtering insufficiently grounded records, merging duplicates, aggregating independent support, preserving unresolved conflicts, abstracting recurrent patterns, and tracking whether retrieved knowledge improves or harms later decisions. The lifecycle comprises:

\begin{figure}[h]
\begin{tikzpicture}[
    node distance=4mm and 4mm,
    stage/.style={
        rounded corners=2pt,
        draw=blue!50!black,
        fill=blue!4,
        align=center,
        font=\scriptsize,
        text width=2.9cm,
        minimum height=0.82cm,
        inner sep=2.5pt
    },
    humanstage/.style={
        stage,
        draw=orange!65!black,
        fill=orange!7
    },
    memorystage/.style={
        stage,
        draw=green!45!black,
        fill=green!5
    },
    flow/.style={
        -{Latex[length=1.6mm]},
        thick,
        draw=black!65
    },
    feedback/.style={
        -{Latex[length=1.6mm]},
        thick,
        dashed,
        draw=green!45!black
    }
]

\node[stage] (capture) {\textbf{1. Experience\\Capture}};
\node[humanstage, right=of capture] (contest) {\textbf{2. Human\\Contestation}};
\node[stage, right=of contest] (reflect) {\textbf{3. Post-Prediction\\Reflection}};
\node[stage, right=of reflect] (verify) {\textbf{4. Verification\\and Staging}};

\node[memorystage, below=8mm of capture] (retrieve) {\textbf{7. Reuse and Transfer}};
\node[memorystage, right=of retrieve] (promote) {\textbf{6. Long-Term\\Promotion}};
\node[memorystage, right=of promote] (consolidate) {\textbf{5. Multi-Experience\\Consolidation}};

\draw[flow] (capture) -- (contest);
\draw[flow] (contest) -- (reflect);
\draw[flow] (reflect) -- (verify);
\draw[flow] (verify.south) |- (consolidate.east);
\draw[flow] (consolidate) -- (promote);
\draw[flow] (promote) -- (retrieve);

\draw[feedback]
    (retrieve.west) -- ++(-0.45,0)
    |- node[pos=0.22,left,font=\scriptsize,align=center] {Future\\cases}
    (capture.west);

\end{tikzpicture}
\end{figure}

\begin{enumerate}[
    label=\textbf{\arabic*.},
    leftmargin=*,
    itemsep=2pt,
    topsep=3pt
]
    \item \textbf{Experience Capture:} record trajectories, outcomes, evidence, provenance, and human or environmental feedback.

    \item \textbf{Human Contestation:}
    allow a human reviewer to challenge and revise the system's reasoning, evidence, or decision through supported contestation actions, while preserving an explicit record of what was contested and how the case state changed.

    \item \textbf{Post-Prediction Reflection:} derive case-specific observations, failure explanations, and candidate patterns that may generalize beyond the current experience.
    
    \item \textbf{Verification and Staging:} assess the grounding, correctness, novelty, and applicable scope of each candidate before retaining it in short-term memory.
    
    \item \textbf{Multi-Experience Consolidation:} group related candidates, merge duplicates, aggregate independent support, preserve conflicting observations, and abstract recurrent knowledge.
    
    \item \textbf{Long-Term Promotion:} promote sufficiently supported and stable knowledge while isolating uncertain or conflicting records for further evidence or review.
    
    \item \textbf{Reuse and Transfer:} retrieve relevant consolidated knowledge for later cases and record where it influences planning, evidence acquisition, argument construction, validation,
and subsequent outcomes.
\end{enumerate}

For multimedia verification, this lifecycle requires a strict epistemic boundary between evidence about the current case and experience derived from previous cases. SEMV enforces this boundary formally in Def.~\ref{def:admissibility} and instantiates the lifecycle in Sec.~\ref{subsec:self_evolution}.
\section{Preliminaries}\label{sec:problem}

\subsection{Problem Definition}
\paragraph{Verification Instance and Output}
Let a multimedia verification instance be:
\begin{equation}
  x=\left\langle
  \Media_x,c_0,\ctx,\Evid^{(0)},\delta \right\rangle,
  \label{eq:verification_instance}
\end{equation}
where $\mathcal{V}_x=\{m_j\}_{j=1}^{n_x}$ contains one or more images or videos. $c_0$ denotes the main textual claim, $\ctx$ contains optional contextual metadata, $\Evid^{(0)}$ is optional user- or dataset-supplied evidence, and $\delta$ identifies the dataset adapter and its output contract.

The system returns:
\begin{equation}
  F(x)=
  \left\langle
  \hat y,\mathbf{p},\gamma,\Just,\Trace,\Xi
  \right\rangle,
  \label{eq:verification_output}
\end{equation}
where $\mathbf{p}\in\Delta(\mathcal{Y}_{\delta})$ is the predicted probability distribution over the dataset-compatible label space and $\hat y=\arg\max_{y} \mathbf{p}(y)$ is the predicted label. The value $\gamma\in[0,1]$ is the confidence assigned to $\hat y$ and is used for calibration, selective prediction, and escalation, which is computed by
Eq.~\ref{eq:final_confidence}. The justification $\mathcal{J}$ is human- and machine-readable, $\mathcal{T}$ is an auditable execution trace, and $\Xi$ contains case-level uncertainty and escalation codes.

For COSMOS, let:
\begin{equation}
\mathcal{Y}_{\mathrm{COSMOS}}=\{\mathsf{NOOC},\mathsf{OOC}\},
\qquad
p_i=\mathbf{p}_i(\mathsf{OOC}),
\qquad
b_i=
\begin{cases}
1, & y_i=\mathsf{OOC},\\
0, & y_i=\mathsf{NOOC},
\end{cases}
\label{eq:cosmos_labels}
\end{equation}
where $p_i$ is the predicted probability of the OOC class and $b_i\in\{0,1\}$ is the corresponding binary ground-truth indicator. 
The mapping from the internal argumentative score to $p_i$ is given in Eq.~\ref{eq:prob_map}

For MV2026, $\mathcal{Y}_{\mathrm{MV}}$ contains the adapter-defined labels for verified, mostly verified, partially verified, false-context, manipulated or synthetic, and insufficient-evidence outcomes. Because uncertainty is an abstention outcome rather than an official task label, it is represented by $\Xi$ and is not included in $\mathcal{Y}_{\mathrm{MV}}$.

\paragraph{Case Evidence and Arguments}
After input canonicalization, the supplied evidence $\Evid^{(0)}$ is incorporated into the initial current-case evidence set $\Evid_0$. Subsequent retrieval and analysis produce the evolving evidence set $\Evid_t$.

An evidence item $e$ records its identifier, source, observed content, reliability $r_e$, relevance $q_e$, media or temporal location $\lambda_e$, provenance $\rho_e$, and uncertainty codes $\unc_e$. An argument $a$ records its identifier, associated subclaim, stance $\sigma_a$, text $t_a$, cited evidence identifiers $E(a)$, cited memory identifiers $M(a)$, grounding-verifier outcome $v_a$, strength $w_a$, and uncertainty codes $\unc_a$:
\begin{equation}
\begin{aligned}
e&=\left\langle \mathrm{id}_e,\mathrm{src}_e,\mathrm{obs}_e,r_e,q_e,
\lambda_e,\rho_e,\unc_e\right\rangle,\\
a&=\left\langle \mathrm{id}_a,k,\sigma_a,t_a,E(a),M(a),v_a,w_a,\unc_a\right\rangle,
\end{aligned}
\label{eq:case_objects}
\end{equation}
where $k$ indexes a scoped subclaim,
$v_a\in\{0,1\}$, $w_a\in[0,1]$, and $\sigma_a\in\{\mathsf{sup},\mathsf{att},\mathsf{mix},\mathsf{neu}\}$, corresponding to \textit{support},\textit{ attack}, \textit{mixed}, and \textit{neutral} positions. The output-level codes $\Xi$ aggregate unresolved or decision-relevant codes from the local sets $\unc_e$ and $\unc_a$.

Because $E(a)$ stores identifiers rather than evidence objects, define $\IDs(\Evid_t)=\{\mathrm{id}_e:e\in\Evid_t\}$ and let the evidence resolved for argument $a$ at state $s_t$ be:
\begin{equation}
  \Evid_t(a)=\left\{e\in\Evid_t:\mathrm{id}_e\in E(a)\right\}.
  \label{eq:resolved_argument_evidence}
\end{equation}

\begin{definition}[Epistemic admissibility]
\label{def:admissibility}
At state $s_t$, an argument $a$ is epistemically admissible if and only if it cites at least one current-case evidence item, every cited identifier resolves in the current evidence set, every resolved item has provenance, and the grounding verifier
accepts the evidence--argument relation:
\begin{equation}
\begin{split}
\operatorname{Adm}_t(a)=1 \iff {}
E(a)\neq\varnothing
\ \wedge\
E(a)\subseteq\operatorname{IDs}(\mathcal{E}_t){}\wedge\
\bigwedge_{e\in\mathcal{E}_t(a)}
[\rho_e\neq\bot]
\ \wedge\
[v_a=1].
\end{split}
\label{eq:admissibility}
\end{equation}
Retrieved memory may guide planning and argument construction, but $M(a)$ cannot replace $E(a)$ in Eq.~\ref{eq:admissibility}.
\end{definition}

This requirement links every admissible argument to current-case evidence, which is what makes provenance validation, human contestation, and execution replay possible. Arguments that fail are not discarded. They remain visible in the trace and are marked inadmissible as verified evidence. If retained for reasoning, they contribute only with the grounding-aware discounted weighting defined in Eq.~\ref{eq:argument_strength}, preserving diagnostic information without assigning them verified status.

\subsection{Multi-Agent Transition System}
We model SEMV as a cooperative, partially observable multi-agent transition system:
\begin{equation}
\mathfrak{S}=\left\langle
  \mathcal{A},\mathcal{S},\{\mathcal{U}_i\}_{i\in\mathcal{A}},
  \mathcal{L},\{\Omega_i,\pi_i\}_{i\in\mathcal{A}},
  \mathcal{P},T,\mathcal{I}
\right\rangle.
\label{eq:mas}
\end{equation}
where $\Agents=\{A_1,\ldots,A_n\}$ is the set of functional agents, $\mathcal{S}$ is the shared case-state space, $\mathcal{U}_i$ is the action set of agent $i$, and $\mathcal{L}$ is the set of structured inter-agent messages. The local observation function $\Omega_i:\mathcal{S}\rightarrow\mathcal{O}_i$ exposes only the fields needed by agent $i$, while $\pi_i:\mathcal{O}_i\times\mathcal{L}^{*}\rightarrow\Delta(\mathcal{U}_i)$ selects an action from that local view and the received message history. The dependency graph $\mathcal{P}$ specifies forward order and downstream invalidation, $T$ governs state transitions, and $\mathcal{I}$ contains the invariants that every committed state must satisfy.

The global state after logical step $t$ is
\begin{equation}
\begin{split}
s_t=\langle {}B_t,\mathcal{C}_t,\mathcal{E}_t,G^E_t,
\mathcal{R}_t,\mathcal{G}_t,\{Q_{k,t}\}_{k\in\mathcal{C}_t},
&D_t,Y_t,\mathcal{M}_t,\mathcal{H}_t,\Xi_t,\mathcal{T}_t\rangle.
\end{split}
\label{eq:global_state}
\end{equation}
containing the canonical case bundle $B_t$, scoped claims $\mathcal{C}_t$, normalized evidence $\mathcal{E}_t$, evidence graph $G^E_t$, research plans $\mathcal{R}_t$, arguments $\mathcal{G}_t$, claim-level A-QBAFs $Q_{k,t}$, subclaim and final decisions $D_t$ and $Y_t$, memory view $\mathcal{M}_t$, human reviews $\mathcal{H}_t$, uncertainty codes $\Xi_t$, and append-only trace $\mathcal{T}_t$. 

For $u\in\mathcal{U}_i$, let $\operatorname{Pre}_i(s_t,u)$ denote the stage prerequisites of agent $i$. The transition semantics are 
\begin{equation}
T(s_t,i,u)=
\begin{cases}
s', & \operatorname{Pre}_i(s_t,u)\ \wedge\ s'\models\mathcal{I},\\
\bot, & \text{otherwise},
\end{cases}
\label{eq:transition}
\end{equation}
where $\bot$ is a rejected update. An agent action is therefore not committed merely because it returns a syntactically valid object; it must also preserve the system-wide constraints.

\begin{definition}[Valid case state]
A state $s$ is valid, written $s\models\mathcal{I}$, when it jointly satisfies
\begin{equation}
s\models\mathcal{I}\iff{}I_{\mathrm{id}}(s)\wedge
I_{\mathrm{prov}}(s)\wedge I_{\mathrm{arg}}(s)\wedge
I_{\mathrm{qbaf}}(s)
\wedge I_{\mathrm{trace}}(s)\wedge I_{\mathrm{gold}}(s)
\wedge I_{\mathrm{freeze}}(s)\wedge I_{\mathrm{out}}(s).
\label{def:valid_state}
\end{equation}
The constraints require unique case-local identifiers, provenance for normalized evidence, correct admissibility status and grounding-aware discounted weighting for every argument entering an A-QBAF, and existing endpoints for every A-QBAF edge.
They also require the new trace to extend rather than overwrite the previous trace, prohibit access to hidden supervision before prediction, prohibit memory mutation during frozen evaluation, and require the final label to belong to the active dataset contract.
\end{definition}

State updates are committed sequentially to preserve these invariants. Independent extraction, retrieval, and claim-level reasoning operations execute concurrently, but their results enter the shared state only after validation.
Table~\ref{tab:agents} in the Supplementary Material summarizes the agent roles.

\section{Methodology}\label{sec:method}
Fig.~\ref{fig:semv_architecture}a divides SEMV into six modules. Modules A--E verify the current case by standardizing inputs, extracting multimodal evidence, decomposing and researching claims, constructing evidence-grounded bipolar arguments, and producing a contestable decision. Module F runs only after eligible supervision becomes available, consolidating grounded experience into reusable external memory without updating model parameters.

\begin{figure}[ht!]
    \centering
    \includegraphics[width=.95\linewidth]{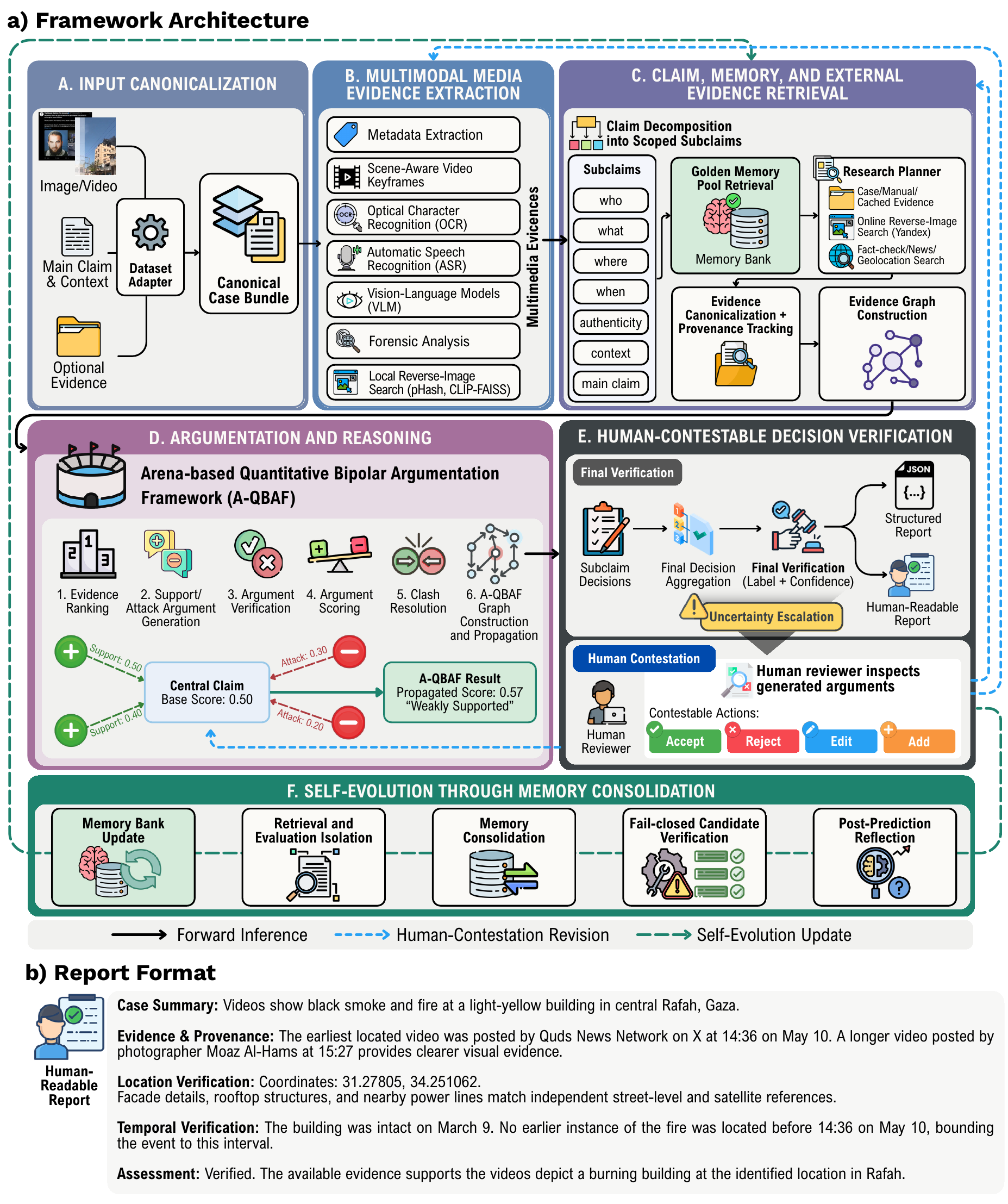}
    \caption{\textit{SEMV framework and report format.} a) End-to-end architecture for multimedia verification, human contestation, and memory-based self-evolution. Solid, blue dashed, and green dashed arrows denote forward inference, contestation-driven revision, and self-evolution updates, respectively. b) Human-readable report format summarizing the case, evidence and provenance, location and temporal verification, and final assessment.}
    \label{fig:semv_architecture}
\end{figure}

\subsection{Module A: Input Canonicalization}
\label{subsec:canonicalization}
Multimedia benchmarks differ in folder layout, metadata fields, label spaces, and reporting requirements. Module A isolates these differences through a dataset-specific adapter $B_x=\mathsf{Can}_{\delta}(x)$.
The resulting bundle contains the normalized textual input, media manifest, source groupings, available time and location context, optional supplied claims and evidence, and run and output settings.

For each asset, the adapter assigns a stable case-local identifier and records its media type, case role, path, source group, sequence position, language hints, and available context. Paths and schema fields are validated before execution, while content hashes are added later by the metadata extractor when the files are available. The adapter also translates the internal decision into the benchmark label space. Evaluation-only annotations remain inaccessible during inference and can be exposed to reflection only after the first-pass report is committed. Canonicalization therefore provides both schema interoperability and the leakage boundary for subsequent reasoning.

\subsection{Module B: Multimodal Media Evidence Extraction}\label{subsec:media_extraction}
Module B converts each canonical asset into localized, provenance-bearing observations. Let $\Ext(m_j)$ denote the enabled extractors applicable to asset $m_j$. The complete media-derived evidence set is:
\begin{equation}
  \Evid^{\mathrm{media}}_x=
  \bigcup_{m_j\in\Media_x}\ \bigcup_{g\in\Ext(m_j)}g(m_j;c_0,\ctx).
  \label{eq:media_evidence_union}
\end{equation}
The extractor family covers metadata inspection, scene-aware keyframe extraction, optical character recognition (OCR), automatic speech recognition (ASR), VLM-based visual interpretation, forensic analysis, and local reverse-image search; the specific tools, versions, and per-tool thresholds are listed in Sec.~\ref{app:tools} in Supplementary Material to keep this section focused on the verification logic.

For images, the original asset serves as the visual target; a SHA-256 digest supports integrity checking and duplicate detection, and the VLM returns structured observations covering scene, salient objects, visible people, embedded text, and location, time, event, and authenticity cues. For videos, scene boundaries are detected and the temporal midpoint of each scene is decoded, retaining at most $K_f$ representative frames ($K_f=8$ by default); when scene detection yields no usable boundaries, timestamps are sampled uniformly. Near-duplicate frames are removed by perceptual-hash distance, and the audio stream is transcribed with segment timestamps. Missing dependencies, unreadable media, and empty tool outputs are represented as uncertainty items rather than being interpreted as evidence for manipulation or false context. All tools return the common evidence representation in Eq.~\ref{eq:case_objects}.

\subsection{Module C: Claim, Memory, and External Evidence Retrieval}
\label{subsec:claim_memory_retrieval}

\paragraph{Scoped Claim Decomposition}
A broad caption may contain several independently verifiable statements. To make failures localizable, the decomposition agent produces:
\begin{equation}
  \mathcal{C}_x=\mathsf{Dec}(B_x,\mathcal{E}^{\mathrm{media}}_x)
  =\{c_k\}_{k=1}^{K}.
  \label{eq:claim_decomposition}
\end{equation}
Each subclaim stores an identifier, statement, claim type, initial queries, and case, media, segment, source, or event scope. The default dimensions are \emph{what}, \emph{where}, \emph{when}, \emph{who}, \emph{why}, and \emph{authenticity}. Deterministic fallbacks complete missing or malformed model output so that the configured dimensions remain available.

\paragraph{Retrieval from Consolidated Memory}
For each $c_k$, the memory agent searches only active long-term records. The query combines the claim, task scope, case context, media observations, temporal signals, and geolocation cues. Records are ranked by a weighted combination of semantic similarity, task and scope compatibility, agreement with the observed evidence pattern, consolidated confidence, independent support, and prior successful use; the weights are listed in Sec.~\ref{app:constants} in Supplementary Material. A record must pass separate semantic and final-score gates. Equivalent records are deduplicated before the top $K_m$ candidates are returned ($K_m=5$ by default). Retrieved memory can change the verification questions, search queries, and reasoning checks, but by Def.~\ref{def:admissibility} it is never inserted into the current-case evidence set.

\paragraph{Research Planning and External Retrieval}
The research planner turns each subclaim into verification questions, search queries, preferred source types, and uncertainty checks, enriched with OCR text, ASR phrases, visual cues, dates or coordinates, reverse-search matches, geolocation candidates, and any applicable memory guidance. The retriever combines supplied or cached evidence with enabled web, news, fact-checking, geolocation, and online reverse-image services. A retrieval gap produces an uncertainty observation rather than a fabricated source. The trace distinguishes memory that was merely retrieved from memory that actually influenced a plan.
The resulting evidence pool is:
\begin{equation}
  \mathcal{E}_x=\mathsf{Norm}\!\left(
  \mathcal{E}^{(0)}\cup\mathcal{E}^{\mathrm{media}}_x
  \cup\bigcup_{k=1}^{K}\mathcal{E}^{\mathrm{ext}}_k\right).
  \label{eq:combined_evidence}
\end{equation}
Normalization completes provenance records, merges repeated identifiers, retains the strongest available reliability and relevance values for a duplicate, and combines uncertainty codes. The evidence graph then links subclaims to usable evidence and evidence to source-provenance records. Before argument construction, evidence is ranked per subclaim, human-rejected items are excluded, and the ten highest-ranked items are passed to the argumentation agents by default.

\subsection{Module D: Argumentation and Reasoning}
\label{subsec:argumentation}

\paragraph{Proponent and Opponent Construction}
For each subclaim, proponents and opponents construct support and attack arguments from the same ranked evidence pool. Every generated argument must return current-case evidence identifiers and may separately disclose memory identifiers that shaped its construction. Unknown evidence identifiers are removed or repaired only when they resolve to a current-case item; an argument with no resolvable evidence fails admissibility. The argument verifier then assesses whether the linked evidence supports the argument text. Missing evidence, malformed verifier output, or verifier unavailability sets $v_a=0$ and records a uncertainty code.

\paragraph{Argument Strength}
For the resolved evidence $\Evid_t(a)$, let $r_a$ and $q_a$ denote mean source reliability and claim relevance. Let $\dvs^{\mathrm{src}}_a=\min(n^{\mathrm{src}}_a/3,1)$ and $\dvs^{\mathrm{mod}}_a=\min(n^{\mathrm{mod}}_a/2,1)$ measure source and modality diversity, where $n^{\mathrm{src}}_a$ and $n^{\mathrm{mod}}_a$ count distinct sources and modalities. Let $\prov_a$ be the fraction of resolved items carrying provenance and $\spec_a=\clip(|t_a|/280,0.25,1)$ a length-normalized specificity proxy. The strength passed to the A-QBAF is:
\begin{equation}
w_a=\clip\bigl(\eta_v(a)\,\eta_u(a)\bigl[
0.20\,r_a+0.25\,q_a+0.15\,\dvs^{\mathrm{src}}_a
+0.15\,\dvs^{\mathrm{mod}}_a+0.20\,\prov_a+0.05\,\spec_a\bigr],0,1\bigr),
\label{eq:argument_strength}
\end{equation}
where $\eta_v(a)=1$ for a verified argument and $0.25$ otherwise, and $\eta_u(a)=0.85$ when uncertainty is present and $1$ otherwise. Each component lies in $[0,1]$, so $w_a\in[0,1]$ by construction, and the clip is defensive only. The score rewards relevant, well-provenanced, independently corroborated arguments. We use \textit{grounding-aware discounted weighting}, whereby arguments that fail verification remain traceable but receive reduced influence through $\eta_v(a)$ rather than being treated as verified evidence. This preserves diagnostic information while distinguishing provisional influence from admissible support. We note that $\spec_a$ is a proxy that rewards argument length rather than informational specificity directly. It carries the smallest weight in Eq.~\ref{eq:argument_strength}, and we report a sensitivity analysis over all six weights in
Sec.~\ref{subsec:sensitivity} in the Supplementary Material.

\paragraph{A-QBAF Construction and Propagation}
We adopt A-QBAF as the reasoning substrate because its support/attack relations and quantitative argument strengths provide an editable decision state through which accepted interventions can be propagated to the resulting judgment \cite{cao2026neuro}.

For each subclaim $c_k$, the reasoning agent constructs the claim-centered QBAF:
\begin{equation}
  Q_k=\langle V_k,R_k^{+},R_k^{-},\beta_k\rangle,
  \qquad V_k=\{c_k\}\cup\mathcal{G}_k,
  \label{eq:qbaf}
\end{equation}
where $R_k^{+}$ and $R_k^{-}$ contain support and attack edges from arguments to the central subclaim, and the neutral base value is $\beta_k(c_k)=0.5$. Mixed and neutral arguments remain in the audit trace but do not create bipolar edges. The total support, total attack, and net argumentative effect are:
\begin{equation}
  S_k=\sum_{(a,c_k)\in R_k^{+}}w_a, \quad
  A_k=\sum_{(a,c_k)\in R_k^{-}}w_a, \qquad
  \Delta_k=S_k-A_k.
  \label{eq:qbaf_effect}
\end{equation}

Using the saturating influence function $\infl(\nu)=\max(\nu,0)^2/[1+\max(\nu,0)^2]$, the propagated claim score is:
\begin{equation}
  \score_k=\clip\!\left(
  \beta_k+(1-\beta_k)\,\infl(\Delta_k)-\beta_k\,\infl(-\Delta_k),0,1\right).
  \label{eq:qbaf_propagation}
\end{equation}

Since $V_k$ contains only the central subclaim and its direct arguments, $Q_k$ forms a star graph to keep each argument directly linked to both its current-case evidence (Def.~\ref{def:admissibility}) and the subclaim it informs. Conflicting evidence is represented as competing support and attack on the same subclaim, providing a compact and traceable case representation. Accordingly, Eq.~\ref{eq:qbaf_propagation} applies the gradual semantics to this structure.

Two properties follow directly and are stated here because they govern the interpretation of the decision thresholds introduced below. First, with $\beta_k=0.5$ and $\infl\in[0,1)$, the map is strictly monotone in $\Delta_k$ and confined to $(0,1)$, so the outer clip never binds. Second, inverting Eq.~\ref{eq:qbaf_propagation} shows that the decision bands correspond to fixed thresholds on the \emph{net} evidence weight, with $\score_k>0.70$ requiring $\Delta_k>0.816$ and $\score_k>0.55$ requiring $\Delta_k>0.333$. A single verified argument with maximal component scores attains $w_a=1$, so one such argument is sufficient to reach the supported band. The function also saturates quickly ($\Delta_k=2\Rightarrow\score_k=0.90$;$\Delta_k=5\Rightarrow\score_k=0.98$), and $\score_k$ therefore carries little resolution beyond $\Delta_k\approx3$. 
When the strongest support and attack both reach $0.55$ and differ by at most $0.25$, the clash resolver flags a major factual conflict, reduces the competing argument scores, and re-propagates the graph. The score is converted to a subclaim decision by:
\begin{equation}
d_k=
\begin{cases}
\mathsf{refuted}, & \score_k<0.30,\\
\mathsf{weakly\text{-}refuted}, & 0.30\leq \score_k<0.45,\\
\mathsf{uncertain}, & 0.45\leq \score_k\leq0.55,\\
\mathsf{weakly\text{-}supported}, & 0.55<\score_k\leq0.70,\\
\mathsf{supported}, & \score_k>0.70.
\end{cases}
\label{eq:subclaim_mapping}
\end{equation}

\paragraph{Final Decision Aggregation}
The aggregator applies a priority-ordered set of rules. Refuted authenticity yields \emph{manipulated or synthetic}. Supported event content together with a refuted where, when, who, or caption-context claim yields \emph{false context}, which the COSMOS adapter maps to \emph{OOC}. If all present core dimensions (what, where, when, who, and authenticity) are supported, the result is \emph{verified} when their mean score is at least $0.70$ and \emph{mostly verified} otherwise. A weakly supported subclaim yields \emph{partially
verified}. If at least $\max(1,\lfloor K/2\rfloor)$ subclaims are uncertain, the result is \emph{uncertain}. A case whose mean subclaim score falls below $0.45$ without triggering any preceding rule is reported as \emph{insufficient evidence}; any remaining unresolved case is \emph{uncertain}.

Let $\bar{\score}=K^{-1}\sum_{k=1}^{K}\score_k$ be the mean subclaim score. The internal escalation confidence is:
\begin{equation}
\gamma^{\mathrm{esc}}=\max\!\left\{2\left|\bar{\score}-0.5\right|,\
\cfloor{\tilde{y}}\right\},
\label{eq:final_confidence}
\end{equation}
where $\tilde{y}$ is the internal decision and $\cfloor{\tilde{y}}$ a label-specific floor preventing a decisive rule from being paired with a near-zero confidence (values in Sec.~\ref{app:constants} in the Supplementary Material).

For binary adapters, the class probability required by Eq.~\ref{eq:cosmos_labels} is obtained from the mean subclaim score rather than from the escalation confidence:
\begin{equation}
p_i=\mathbf{p}_i(\mathsf{OOC})=
\begin{cases}
\bar{\score}_i, & \text{internal decision is false-context},\\
1-\bar{\score}_i, & \text{otherwise}.
\end{cases}
\label{eq:prob_map}
\end{equation}
The adapter then maps the internal label to $y\in\Labels_{\delta}$ and the report agent renders the result, evidence links, arguments, uncertainty reasons, and trace into the human-readable format illustrated in Fig.~\ref{fig:semv_architecture}b.

\subsection{Module E: Human-Contestable Decision Verification}
\label{subsec:contestation}

\paragraph{Uncertainty Detection}
Before a decision is presented as final, the uncertainty detection module checks each subclaim for conditions that warrant review. These include a score in the neutral band, missing or uniformly low-reliability evidence, OCR--ASR disagreement, metadata--claim conflict, an earlier-context reverse-search match, forensic--source conflict, and a near tie between support and attack. The system preserves the assigned decision together with the reason codes and affected stages, so a reviewer can therefore distinguish factual conflict from a missing tool, missing evidence, or insufficient retrieval.

\paragraph{Human Contestation as Causal Intervention}

\begin{figure}[htbp]
\centering
\begin{minipage}[t]{0.49\textwidth}
\vspace{0pt}
\centering
\resizebox{\linewidth}{!}{
\begin{tikzpicture}[
    stage/.style={draw=SEMVNavy,rounded corners=2pt,fill=PaleBlue,
        minimum width=22mm,minimum height=8mm,font=\scriptsize\bfseries,align=center},
    humanbox/.style={draw=SEMVTeal,rounded corners=2pt,fill=PaleGreen,
        minimum width=28mm,minimum height=7mm,font=\scriptsize\bfseries,align=center},
    arr/.style={-{Latex[length=2mm,width=1.3mm]},thick,draw=SEMVNavy,
        shorten <=1pt,shorten >=1pt},
    human/.style={-{Latex[length=2mm,width=1.3mm]},dashed,thick,draw=SEMVTeal,
        shorten <=1pt,shorten >=1pt},
    routelabel/.style={font=\scriptsize,text=SEMVTeal,fill=white,
        inner xsep=1.5pt,inner ysep=0.8pt,align=center}]
\node[stage] (c) at (0,0) {Claims $\stC$};
\node[stage] (r) at (2.7,0) {Retrieve $\stR$};
\node[stage] (v) at (5.4,0) {Validate $\stV$};
\node[stage] (a) at (5.4,-3.2) {Arguments $\stA$};
\node[stage] (q) at (2.7,-3.2) {A-QBAF $\stQ$};
\node[stage] (f) at (0,-3.2) {Aggregate $\stF$};
\node[stage,draw=orange!75!black,fill=orange!18] (o) at (0,-1.64) {Report $\stO$};
\draw[arr] (c.east) -- (r.west);
\draw[arr] (r.east) -- (v.west);
\draw[arr] (v.south) -- (a.north);
\draw[arr] (a.west) -- (q.east);
\draw[arr] (q.west) -- (f.east);
\draw[arr] (f.north) -- (o.south);
\node[humanbox, minimum height=23pt] (h) at (2.7,-1.64) {Human action $h$};
\draw[human] (h.north) -- node[routelabel,pos=0.50,left=1pt] {add source} (r.south);
\draw[human] (h.north east) -- node[routelabel,above,pos=0.29] {reject evidence} (5,-0.4);
\draw[human] (h.south east) -- node[routelabel,below,pos=0.35] {edit argument} (4.98,-2.8);
\draw[human] (h.south) -- node[routelabel,pos=0.50,left=1pt] {accept} (q.north);
\end{tikzpicture}}
\captionof{figure}{\textit{Localized human contestation.} Dashed arrows route each human action to the earliest affected stage $u^{*}$.}
\label{fig:revision}
\end{minipage}
\hfill
\begin{minipage}[t]{0.49\textwidth}
\vspace{0pt}
\centering
\captionof{table}{\textit{Default long-term memory promotion requirements.}
Thresholds correspond to $(\theta_\tau^{c},\theta_\tau^{x},\theta_\tau^{s})$ in
Eq.~\ref{eq:memory_promotion}.}
\label{tab:memory_thresholds}
\small
\renewcommand{\arraystretch}{1.18}
\resizebox{.92\linewidth}{!}{
\begin{tabular}{lccc}
\toprule
\textbf{Memory type} &
$\boldsymbol{\theta_\tau^{c}}$ &
$\boldsymbol{\theta_\tau^{x}}$ &
$\boldsymbol{\theta_\tau^{s}}$ \\
 & conf. & cases & sources \\
\midrule
Episodic observation & $0.85$ & $1$ & $1$ \\
Failure lesson       & $0.70$ & $2$ & $2$ \\
Semantic rule        & $0.75$ & $3$ & $3$ \\
\bottomrule
\end{tabular}}
\end{minipage}
\end{figure}

Human review is modeled as an intervention on the computation that connects evidence to the final decision. A review action is represented as: 
\begin{equation}
h=\bigl\langle
o_h,\ \mathrm{tgt}_h,\ t_h,\ \varsigma_h,\
E_h^{+},E_h^{\Delta},\ \mathrm{rsn}_h,\ \mathrm{meta}_h
\bigr\rangle,
\label{eq:human_intervention}
\end{equation}
where $o_h\in\{\textsf{accept},\textsf{reject},\textsf{edit},\textsf{add}\}$ is available contestation action, $\mathrm{tgt}_h$ identifies the affected argument or subclaim, $t_h$ and $\varsigma_h$ are optional revised text and stance, $E_h^{+}$ and $E_h^{\Delta}$ contain newly supplied and edited evidence identifiers, $\mathrm{rsn}_h$ records the reviewer's reason, and $\mathrm{meta}_h$ stores metadata including an optional explicit revision target. Review therefore modifies the evidence--argument--decision relation rather than attaching an unstructured comment to the final label.

The stage symbols follow the pipeline order $\stC\prec\stR\prec\stV\prec\stA\prec\stQ\prec\stF\prec\stO$ (claims, retrieval, validation, arguments, A-QBAF, aggregation, and output). If no explicit target is provided, our system routes to the earliest pipeline stage that the intervention invalidates. By default:
\begin{equation}
\operatorname{route}(h)=
\begin{cases}
\stQ, & o_h=\textsf{accept},\\
\stR, & \text{$h$ introduces, replaces, or challenges source evidence},\\
\stV, & \text{$h$ challenges the content or validity of existing evidence},\\
\stA, & \text{otherwise}.
\end{cases}
\label{eq:revision_routing}
\end{equation}

Thus, missing, mismatched, or newly supplied sources restart processing from retrieval; challenges involving irrelevant, unsupported, or factually inconsistent evidence restart from validation; other edits, additions, and rejections restart from argument construction; and acceptance requires only A-QBAF reasoning and the subsequent decision stages. A valid explicit target in $\mathrm{meta}_h$ overrides these defaults. Fig.~\ref{fig:revision} illustrates the four routing paths.

For a non-empty review batch $H$, let $\mathcal{C}_H$ denote the subclaims identified directly or through the contested arguments and evidence. The executor selects the earliest required stage and applies the interventions before recomputation:
\begin{equation}
u_H=\min_{\prec}\{\operatorname{route}(h):h\in H\},
\qquad
s^{(H)}
=
T^{*}_{\mathsf{Closure}(H)}
\bigl(\mathsf{Apply}(s,H)\bigr),
\label{eq:scoped_revision}
\end{equation}
where $\Objs^{*}$ is the downstream dependency closure of $u^{*}$ for $\Claims_H$ in the case-specific dependency graph, and $T^{*}_{\Objs^{*}}$ is the ordered composition of the corresponding stage transitions. Retrieval, validation, and argument reconstruction are localized to the affected subclaims, while A-QBAF reasoning, final aggregation, and report generation are recomputed as required to maintain a consistent case-level decision. If the supplied identifiers are insufficient to establish a safe local scope, $\Claims_H$ or $u^{*}$ is expanded conservatively.

Rejected arguments remain in the audit trace but are excluded from the revised A-QBAF. An edit preserves the original argument and introduces a linked replacement, while an added argument carries explicit human provenance. When recomputation begins at argument construction or the earlier, retained human-reviewed arguments are revalidated and rescored against the current evidence state.
Importantly, contestation never overwrites the initial report. The append-only trace stores the review batch, inferred revision plan, affected subclaims and evidence, the first-pass output $F^{(0)}(x)$, the revised output $F^{(H)}(x)$, and their label and confidence differences.

\begin{definition}[Contestation locality]\label{def:locality}
Let \(\mathcal{Z}\) be the set of observable state objects and \(\mathcal{Z}^* \subseteq\mathcal {Z}\) the intended dependency scope of a human contestation \(H\). Contestation locality is the proportion of state objects outside this scope that are preserved after revision:
\begin{equation}
\mathsf{Loc}(H)=1-
\frac{\displaystyle\sum_{z\in\Objs\setminus\Objs^{*}}
\mathbf{1}\!\left[z\!\left(s^{(H)}\right)\neq z(s)\right]}
{\left|\Objs\setminus\Objs^{*}\right|}.
\label{eq:contestation_locality}
\end{equation}
If $\Objs\setminus\Objs^{*}=\varnothing$ we define $\mathsf{Loc}(H)=1$.
\end{definition}

An idealized system with perfect dependency tracking and deterministic transitions attains $\mathsf{Loc}(H)=1$, but that statement is definitional and uninformative about a real implementation, in which dependency tracking is incomplete, state is shared and mutable, retrieval is live, and generation is stochastic. The useful question is how far locality can fall, and why. The following decomposition answers it and makes the measured values in Sub.~\ref{subsec:contest} interpretable.

\begin{proposition}[Locality degradation bound]
\label{prop:locality_bound}
Let $\Objs^{\mathrm{true}}$ denote the true causal closure of $H$, i.e., the set of objects whose values can depend on the intervention, and let $\Objs^{*}$ denote the scope computed by the executor. Define the \emph{untracked-dependency count} as $\delta_H=\left|\Objs^{\mathrm{true}}\setminus\Objs^{*}\right|$ and the \emph{spurious-write count} $\sigma_H$ as the number of objects outside $\Objs^{\mathrm{true}}\cup\Objs^{*}$ whose values change after the intervention. Then:
\begin{equation}
1-\mathsf{Loc}(H)
\leq
\frac{\delta_H+\sigma_H}
{\left|\Objs\setminus\Objs^{*}\right|}.
\label{eq:locality_bound}
\end{equation}
The bound is tight when every untracked dependent changes value.
\end{proposition}

\begin{proof}
Any changed object outside the computed scope $\Objs^{*}$ is either (i) an object in $\Objs^{\mathrm{true}}\setminus\Objs^{*}$ whose dependency was not captured by the computed scope, or (ii) an object outside the true causal closure that nevertheless changes. The first set contains at most $\delta_H$ changed objects, while the second contains $\sigma_H$ by definition. Thus, the number of changed objects outside $\Objs^{*}$ is at most $\delta_H+\sigma_H$. Dividing by $\left|\Objs\setminus\Objs^{*}\right|$ and applying Eq.~\ref{eq:contestation_locality} gives Eq.~\ref{eq:locality_bound}.
Equality holds when every object in $\Objs^{\mathrm{true}}\setminus\Objs^{*}$ changes under the intervention.
\end{proof}

Proposition~\ref{prop:locality_bound} relates observed locality degradation to two sources: dependencies omitted from the computed scope and changes outside the true causal closure. Accordingly, a low value of $\mathsf{Loc}(H)$ can indicate incomplete scope estimation, unintended out-of-scope state changes, or both.

\subsection{Module F: Self-Evolution through Memory Consolidation}
\label{subsec:self_evolution}
Whereas human contestation revises the current case, self-evolution transfers verified experience to subsequent cases without updating model parameters, following the operational definition in Def.~\ref{def:selfevo}. 

\paragraph{Long-term Memory Composition} Long-term memory contains three types of records: case-specific episodic observations, recurring failure lessons, and generalized semantic rules. Each record stores its lesson, memory type, task and claim scope, trigger and evidence patterns, recommended action, confidence, supporting cases and source fingerprints, conflicts, usage outcomes, lifecycle status, and provenance.

\paragraph{Post-Prediction Reflection}
Reflection begins only after the first-pass or contested report has been finalized. In an eligible training or bootstrap run, the system may then reveal a gold label or use human feedback to identify successful strategies and failure modes. Diagnoses include label mismatch, overconfident error, confusion between contextual misuse and media manipulation, retrieval or tool failure, unresolved argumentative clash, temporal or geolocation error, and weak provenance.

Reflection produces candidate episodic observations, failure lessons, and semantic rules. Each candidate $\cand$ records its lesson and applicability scope, grounding evidence identifiers $E_{\cand}$, grounding argument identifiers $A_{\cand}$, source case and source fingerprint, confidence $c_{\cand}$, supervision source, and verification status.

\paragraph{Candidate Verification Gating}
Let $\Learn$ denote cases processed in an eligible training or bootstrap run, and let $\Evid_i$ and $\Args_i$ denote the evidence and arguments in the finalized trace of source case $x_i$. Candidate activation is governed by:
\begin{equation}
\begin{aligned}
\operatorname{Verify}(\cand)=\mathbf{1}\bigl[&
x(\cand)\in\Learn
\ \land\ (|E_{\cand}|+|\Args_{\cand}|>0)\\
&\land\ E_{\cand}\subseteq\IDs(\Evid_i)
\ \land\ \Args_{\cand}\subseteq\IDs(\Args_i)\land\ c_{\cand}\geq\theta_v
\ \land\ \operatorname{Safe}(\cand)
\ \land\ \neg\operatorname{Dup}_{\mathrm{src}}(\cand)\bigr],
\end{aligned}
\label{eq:memory_verification}
\end{equation}
where $x(\cand)=x_i$ is the candidate's source case and $\theta_v=0.60$ by default. The predicate $\operatorname{Safe}(\cand)$ rejects unsupported or overgeneralized lessons, while $\operatorname{Dup}_{\mathrm{src}}(\cand)$ detects a canonical duplicate already supported by the same case or source fingerprint. Equivalent observations from independent cases are retained because they provide additional support.

Candidates derived from validation or test cases cannot pass Eq.~\ref{eq:memory_verification}. Missing or unresolved grounding blocks verification, while an unavailable or malformed safety check places the candidate under review. A verified contradiction is stored as counter-evidence against the related record rather than activated as a competing rule.

\paragraph{Memory Consolidation}
Verified and under-review candidates are persisted in short-term memory, but only verified candidates participate in promotion. Short-term records exceeding their retention period are archived rather than silently deleted. Every 25 processed learning cases, the consolidator groups compatible observations by memory type, task, claim, applicability scope, and semantic relation.

For a cluster $C$, let $\Ind(C)$ be the largest selected subset in which neither a case identifier nor a source fingerprint is repeated. Its consolidated confidence is:
\begin{equation}
\operatorname{Conf}(C)
=\frac{\alpha_C+\alpha_0}{\alpha_C+\beta_C+\alpha_0+\beta_0},
\quad
\alpha_C=\sum_{\cand\in\Ind(C)}c_{\cand},
\quad
\beta_C=\sum_{\cand\in\Ind(C)}(1-c_{\cand}),
\label{eq:cluster_confidence}
\end{equation}

Let $n_C^{\mathrm{case}}$ and $n_C^{\mathrm{src}}$ denote the numbers of independent cases and source fingerprints in $\Ind(C)$. For memory type $\tau$, promotion is defined by:
\begin{equation}
\operatorname{Promote}(C,\tau)=\mathbf{1}\biggl[
\operatorname{Conf}(C)\geq\theta_\tau^{c}
\land n_C^{\mathrm{case}}\geq\theta_\tau^{x}
\land n_C^{\mathrm{src}}\geq\theta_\tau^{s}
\land\bigwedge_{\cand\in\Ind(C)}\operatorname{Verify}(\cand)=1\biggr].
\label{eq:memory_promotion}
\end{equation}
Default thresholds are given in Table~\ref{tab:memory_thresholds}. The confidence bar is highest for episodic observations because they may be promoted on the strength of a single case: with $\theta^{x}_{\mathrm{epi}}=1$, confidence is the only remaining safeguard. Failure lessons and semantic rules require two and three independent cases and sources respectively, so a lower per-cluster confidence bar still yields a higher total evidential burden. Equivalent or mutually entailing observations are merged into the same long-term record.

\paragraph{Confidence and Conflict Resolution}
Each long-term record $m$ maintains support parameters $(\alpha_m,\beta_m)$, initialized to $(1,1)$. Independent support or contradiction of confidence $c$ updates them as:
\begin{equation}
(\alpha_m,\beta_m)\leftarrow
\begin{cases}
(\alpha_m+c,\ \beta_m+1-c),
& \text{support},\\
(\alpha_m+1-c,\ \beta_m+c),
& \text{contradiction},
\end{cases}
\qquad
c(m)=\frac{\alpha_m}{\alpha_m+\beta_m}.
\label{eq:memory_confidence_update}
\end{equation}

Let $\cratio(m)=n_m^{-}/(n_m^{+}+n_m^{-})$ be the conflict ratio computed from independent support and conflict counts. A ratio above $0.30$ moves an active record under review; deprecation requires both observed contradiction and confidence below $0.45$; an under-review record becomes active again only when $\cratio(m)\leq0.20$ and $c(m)\geq0.45$. Promotions, merges, conflicts, reviews, reactivations, deprecations, and archives are appended to the memory trace.

\paragraph{Retrieval and Evaluation Isolation}
Only active long-term records are retrievable by Module C (Sec.~\ref{subsec:claim_memory_retrieval}). The system distinguishes retrieval exposure from causal use by logging whether a record was retrieved, cited in planning, cited in argument construction, used successfully, or subsequently contested. This separation prevents retrieval frequency from being interpreted as beneficial influence.
For validation and test, the complete memory state is frozen and identified by a deterministic hash.

\begin{remark}[Evaluation-state invariance]
\label{rem:frozen}
In frozen mode the invariant $I_{\mathrm{freeze}}$ of Def.~\ref{def:valid_state} rejects candidate staging, case registration, consolidation, and long-term updates, and evaluation usage events are written outside the snapshot. By Eq.~\ref{eq:transition} no committed transition can therefore alter $\Mem^{L}$, so the snapshot present at the start of evaluation is bit-identical to the snapshot at the end. 
\end{remark}

\section{Experimental Design}
\label{sec:experiments}

\subsection{Datasets and Evaluation Protocols}
We summarize the three evaluation sets and the role each split plays, as shown in Table~\ref{tab:datasets}.

\paragraph{COSMOS}
We use the official COSMOS splits and evaluate on the complete 1,700-image test set, which contains one correct and one mismatched caption for each image \cite{aneja2023cosmos}. The official validation split of 41,006 images is used to select prompts, decision bands, strength weights, and memory-policy thresholds. No memory state produced during development is retained for test evaluation. 

Separately, we define a fixed stratified subset of 1,200 cases from the official COSMOS training split as the memory-learning partition. After all configuration constants have been frozen on the development split, the pre-test memory-learning partition is used to construct the consolidated-memory snapshot. When an experiment requires rebuilding memory, the same 1,200 cases and the same ordering are used for every compared configuration.

To reduce near-duplicate leakage, we compare perceptual image hashes and normalized caption fingerprints between the test set and all items included in the retrieval index or pre-test memory-learning partition. All paired comparisons use identical image--caption identifiers. 

Every full-scale SEMV result on COSMOS is computed over all 1,700 test images. Each full-scale condition uses a single decoding seed because repeating all backbone, memory, argumentation, and sensitivity conditions at full scale exceeded our compute budget. To bound run-to-run variation, we additionally rerun every principal condition with three decoding seeds on a fixed stratified 500-case subset and report the resulting standard deviations and ranges in Table~\ref{tab:run_stability} in the Supplementary Materials.

\paragraph{MV2026}
MV2026 provides 50 training cases, 10 validation cases, and 10 live challenge cases, with evaluation based on the quality of multimedia verification reports \cite{dang20262026}. As no official validation label is provided, we reserve 10 cases from the training set as a held-out subset for primary validation, using the remaining 40 for prompt and memory development.

\paragraph{Contestation Trace Revision (CTR)} To evaluate whether human feedback can contest and correct the verification process, we construct the CTR benchmark from the 10 MV2026 validation cases. Each episode presents reviewers with the retrieved evidence, generated arguments, support and attack relations, scores, dependencies, and intermediate decisions, together with the available contestation actions $o_h$. Reviewers identify the defective trace components, specify the required revisions, and provide the expected downstream state after dependency-aware recomputation. CTR evaluates contestation-action selection, error localization, revision correctness, routing to the earliest affected module, intervention locality, and recovery of the final decision and report.

The annotation protocol involves three annotators: 2 graduate researchers with prior multimedia-verification experience independently define the target state $s_i^{*}$ and dependency closure $\Objs_i^{*}$, and 1 senior researcher adjudicates disagreements. The primary annotators are not involved in implementing SEMV and are blinded to replay configuration. Episodes are randomized and system identifiers removed. Each source case contributes four corrupted episodes and one uncorrupted control. The 40 corruptions are balanced across five families (i.e., evidence or provenance, temporal grounding, geolocation, argument stance or strength, and aggregation or report-field errors) with eight instances per family. A rotated assignment gives each case four distinct corruption families.

\begin{table*}[t]
\centering
\caption{Dataset Splits and Evaluation Protocols.}
\label{tab:datasets}
\footnotesize
\setlength{\tabcolsep}{6pt}
\renewcommand{\arraystretch}{1.2}

\begin{tabularx}{\textwidth}{p{4.2cm}p{2.5cm}X}
\toprule
\textbf{Split} &
\textbf{Cases} &
\textbf{Purposes} \\
\midrule

\rowcolor{gray!18}
\multicolumn{3}{l}{\textbf{COSMOS} \cite{aneja2023cosmos}} \\

Train & 161,752 & Retrieval-index construction and source pool for the fixed pre-test memory-learning partition 
\\ 
$\hookrightarrow$ Memory subset  & 1,200 & Consolidated-memory construction and per-setting memory reconstruction after all constants are frozen 
\\ 
\rowcolor{gray!7} Validation & 41,006 & 
\\
\rowcolor{gray!7} $\hookrightarrow$  Development subset & 40,506 & Prompt, decision-band, strength-weight, and memory-policy constants \\
\rowcolor{gray!7} $\hookrightarrow$  Sensitivity/Stability subset & 500 & Configuration sensitivity and three-seed stability analysis with frozen reference configuration \\
\\
Test & 1,700 & Final OOC detection, calibration, selective prediction, and configuration sensitivity evaluation only \\

\midrule
\rowcolor{gray!18}
\multicolumn{3}{l}{\textbf{MV2026} \cite{dang20262026}} \\

Train &
40 &
Prompt and memory development \\

\rowcolor{gray!7}
$\hookrightarrow$ Validation &
10 &
Pre-challenge evaluation and error analysis \\

Challenge Test &
10 &
Final report quality \\

\midrule
\rowcolor{gray!18}
\multicolumn{3}{l}{\textbf{Contestation Trace Revision (CTR)}} \\

Reviewer Evaluation &
10 (50 episodes) &
Contestation-action selection, error localization, revision correctness,
trace repair, and final-decision recovery \\

\bottomrule
\end{tabularx}
\end{table*}

\subsection{Consolidated Memory Construction}
\label{subsec:memory_provenance}
Because the memory component supplies most of the measured gain on COSMOS (Sub.~\ref{subsec:memory}), the origin of every long-term record must be auditable. Memory is learned exclusively from a dedicated pass over COSMOS training cases, in which, each case is processed through the full Modules A--E pipeline, the report is committed, and only then is the gold label revealed to Module F, as required by Eq.~\ref{eq:memory_verification} and the invariant $I_{\mathrm{gold}}$. No validation or test case can contribute a candidate. Before evaluation the memory state is frozen and hashed (Remark~\ref{rem:frozen}).

\paragraph{Frozen Memory Snapshot}
We process a class-balanced, source-stratified subset of
1,200 COSMOS training cases. The memory bank is generated by GLM-4.6V \cite{hong2025glm}, an open-weight VLM that is not among the four evaluation backbones and shares no weights, pretraining corpus, or model family with any of them. The frozen active snapshot contains 29 episodic, 12 failure, and 7 semantic records, with 8 additional conflict records retained under review.

\subsection{Agent Backbones and Baselines}
We compare four end-to-end SEMV runs by substituting one shared multimodal backbone across all agent roles and the VLM analyzer. Gemma4-31B \cite{gemma42026} and InternVL3.5-38B \cite{wang2025internvl3_5} are dense multimodal models; Qwen3.6-35B \cite{qwen36card} is a multimodal mixture-of-experts (MoE) model with 3B activated parameters; and Nemotron-3-Nano-Omni-30B is a hybrid MoE model with approximately 3B active parameters \cite{deshmukh2026nemotron}. All conditions receive the same canonical cases, original images, sampled keyframes, OCR/ASR transcripts, forensic derivatives, retrieved pages, reverse-image results, and cached tool responses. Visual resolution, frame count, thinking budgets, and maximum tool calls are matched.

Published baselines cover discriminative image-text systems, heuristic or boosting methods, and retrieval-augmented verification. Aneja \textit{et al.}  \cite{aneja2023cosmos} evaluated their adaptations on the original 1,700-record test set, whereas La \textit{et al.} evaluated on a 400-record  holdout \cite{la2022leverage}. We therefore group published scores by protocol and compare numerically only within the full-test block.

\paragraph{Computational Environment}
Experiments run on a single NVIDIA DGX Spark system built on the NVIDIA GB10 Grace Blackwell Superchip, comprising an integrated Blackwell GPU, and 128~GB of coherent unified memory. Models are served locally with \texttt{vLLM} 0.24.0 and \texttt{PyTorch} 2.11.0 in BF16. Runs span 20 July--12 August 2026. All cross-condition comparisons replay the same dated retrieval cache.

\section{Results}\label{sec:results}
In this section, we evaluate SEMV across four dimensions covering predictive performance and confidence reliability, verification report quality, self-evolving memory, and human contestability. 

Sec.~\ref{subsec:predictive} and \ref{subsec:report} examine whether SEMV produces accurate decisions with reliable confidence estimates and generates verification reports that are well grounded, temporally consistent, complete, and computationally efficient. Sec.~\ref{subsec:self_evolution} then evaluates self-evolving memory to determine whether verified experiences can be selectively consolidated, retrieved, and reused to improve future verification while limiting harmful transfer. Finally, Sec.~\ref{subsec:contest} assesses contestability by examining whether human corrections are faithfully translated into system updates, propagated through the appropriate causal dependencies, and applied through localized revision without unnecessarily altering already-correct states.

\subsection{Predictive Performance and Reliability}\label{subsec:predictive}
We evaluate SEMV on the COSMOS test set from two perspectives: predictive performance and confidence reliability. We first assess classification effectiveness across backbones and prior baselines, then examine whether the resulting confidence estimates reliably reflect prediction correctness.

\subsubsection{Predictive Evaluation}
We report accuracy, precision, recall, and F1 on the COSMOS test set, treating OOC as the positive class. We additionally report Wilson score 95\% confidence intervals on accuracy, which remain well calibrated for proportions near the extremes and at finite sample sizes, and we test each backbone against the strongest directly comparable baseline (Tran \textit{et al.} \cite{tran2022textual}) with an unpaired two-proportion $z$-test. The test is unpaired because per-record predictions are unavailable for the baseline, so significance is assessed using the reported accuracies and the shared test-set size.

\paragraph{Results} Fig.~\ref{fig:semv_metric_arena} plots the four SEMV backbones against the strongest directly comparable full-test baseline (Tran \textit{et al.} \cite{tran2022textual}).
Averaged over backbones, SEMV exceeds the baseline in accuracy (90.91\% vs.\ 89.10\%), precision (91.21\% vs.\ 87.80\%), and F1 (90.88\% vs.\ 89.30\%) at comparable recall (90.56\% vs.\ 90.80\%). All four variants match or exceed the baseline in accuracy; three improve F1 and Gemma4 matches it exactly. The two strongest backbones reach the same level from different directions: Nemotron3 is the most OOC-sensitive, with the highest recall (92.47\%, $+1.67$~pp) and F1 (91.82\%, $+2.52$~pp), whereas Qwen3.6 offers a precision-oriented operating point, with the highest accuracy (91.88\%, $+2.78$~pp) and precision (93.00\%, $+5.20$~pp) at a comparable recall of 90.59\%. Both gains are statistically significant under an unpaired two-proportion $z$-test ($p=0.009$ and $p=0.006$), and their 95\% accuracy CIs lie above the baseline point estimate. InternVL3.5 improves accuracy to 90.65\% and F1 to 90.62\%, but the accuracy gain is not significant at the 5\% level ($p=0.140$), and Gemma4 is statistically indistinguishable from the baseline ($p=0.820$). SEMV's gains are therefore not tied to a single backbone, though their magnitude and significance are backbone-dependent.

\begin{figure}[h]
    \centering

    \includegraphics[width=0.8\textwidth]{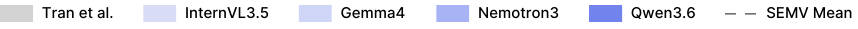}

    \vspace{2pt}
    \begin{subfigure}[t]{0.24\textwidth}
        \centering
        \includegraphics[width=\linewidth]{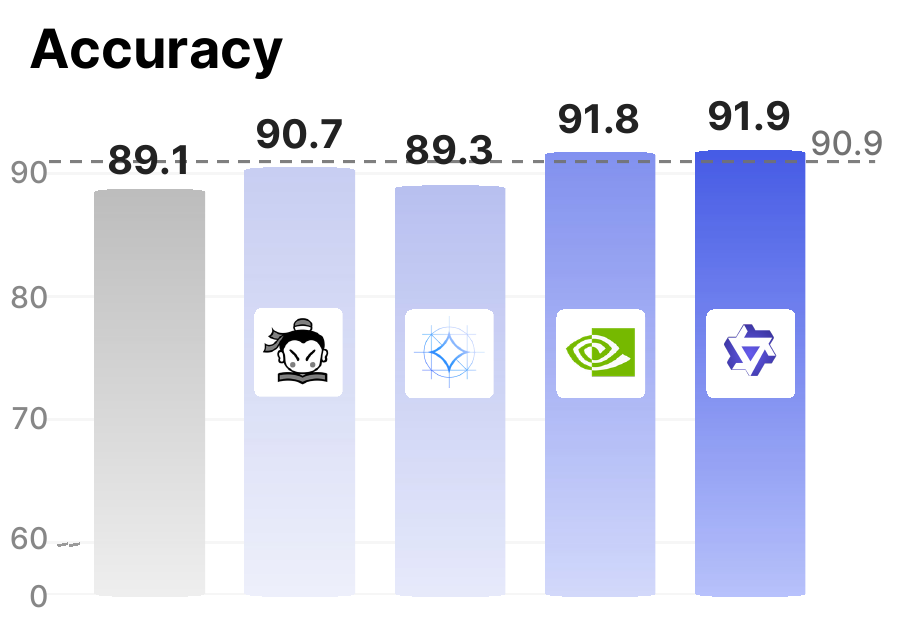}
        \caption{Accuracy}
        \label{fig:semv_metric_accuracy}
    \end{subfigure}
    \begin{subfigure}[t]{0.24\textwidth}
        \centering
        \includegraphics[width=\linewidth]{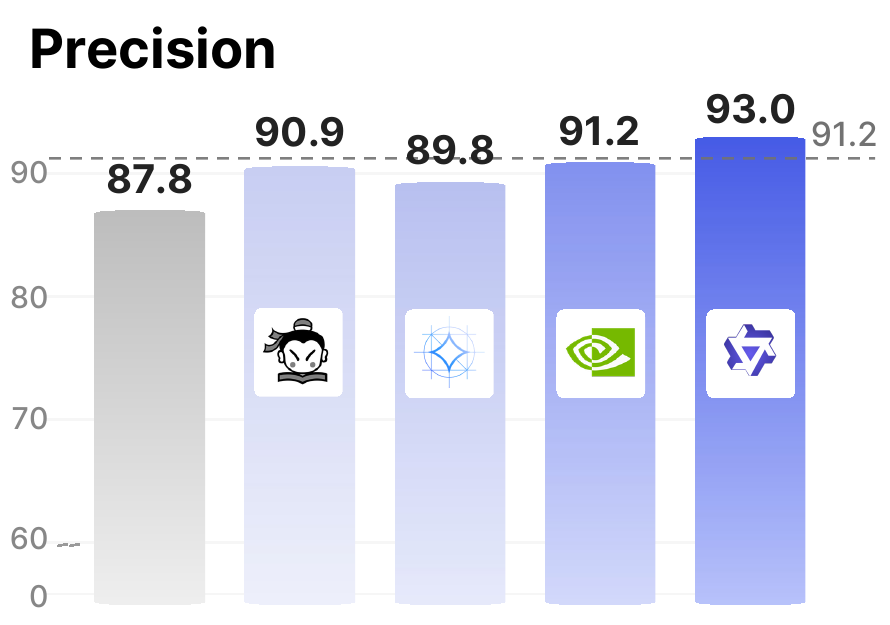}
        \caption{Precision}
        \label{fig:semv_metric_precision}
    \end{subfigure}
    \begin{subfigure}[t]{0.24\textwidth}
        \centering
        \includegraphics[width=\linewidth]{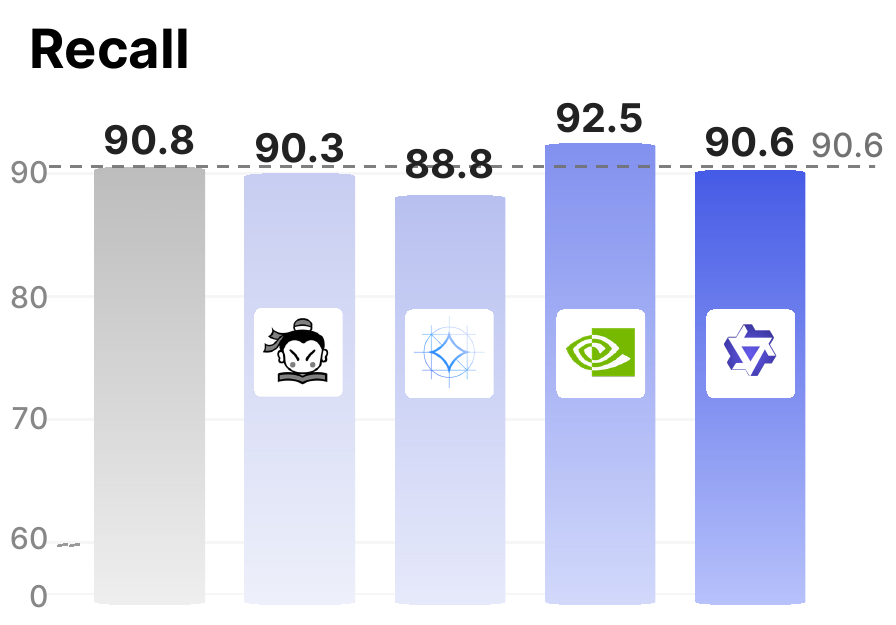}
        \caption{Recall}
        \label{fig:semv_metric_recall}
    \end{subfigure}
    \begin{subfigure}[t]{0.24\textwidth}
        \centering
        \includegraphics[width=\linewidth]{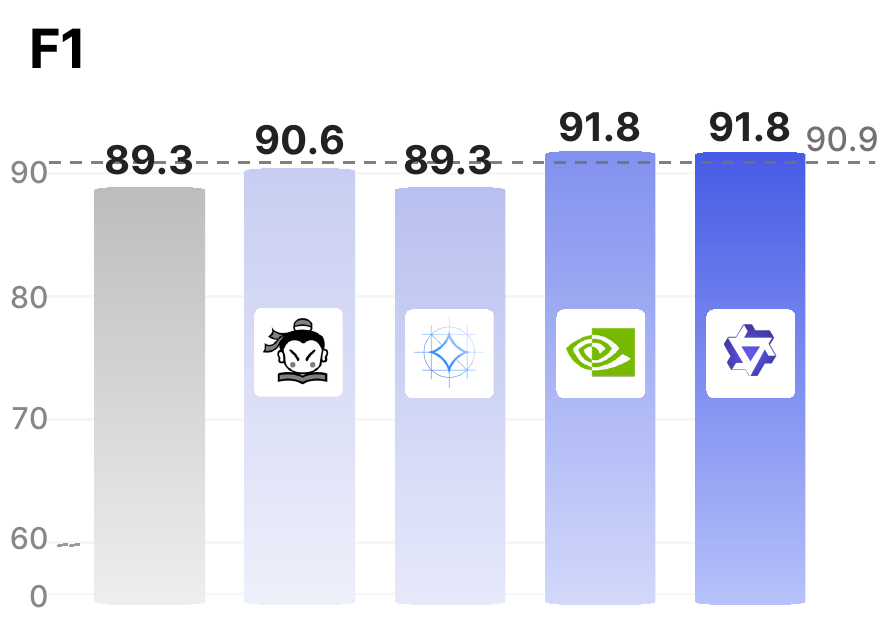}
        \caption{F1}
        \label{fig:semv_metric_f1}
    \end{subfigure}

    \caption{Comparison of SEMV backbones and the best comparable baseline \cite{tran2022textual} on the COSMOS full test set. The dashed line marks the mean score across the four SEMV backbones.}
    \label{fig:semv_metric_arena}
\end{figure}

\begin{table}[h]
\centering
\caption{COSMOS results grouped by evaluation set. The upper block reports prior work on the 400-record held-out subset \cite{la2022leverage}; the lower block reports results on the full 1,700-record test set \cite{aneja2023cosmos}, on which all SEMV backbones are evaluated. Best in \textbf{bold}, second best \underline{underlined}.}
\label{tab:cosmos_prior_results}

\small
\setlength{\tabcolsep}{4pt}
\renewcommand{\arraystretch}{1.12}
\begin{tabularx}{\textwidth}{
    >{\raggedright\arraybackslash}X
    *{4}{>{\centering\arraybackslash}p{1.02cm}}
    *{2}{>{\centering\arraybackslash}p{1.05cm}}
    >{\centering\arraybackslash}p{1.9cm}
    >{\centering\arraybackslash}p{0.85cm}}
\toprule
\textbf{Method} & \textbf{Acc.$\uparrow$} & \textbf{Prec.$\uparrow$}
& \textbf{Rec.$\uparrow$} & \textbf{F1$\uparrow$}
& \textbf{ECE$\downarrow$} & \textbf{Brier$\downarrow$}
& \textbf{95\% CI} & \textbf{$p$} \\
\midrule
\rowcolor{gray!16}
\multicolumn{9}{l}{\textbf{Evaluation set: held-out subset \cite{la2022leverage} ($n=400$)}} \\
SpotFake \cite{singhal2019spotfake} & 53.50 & 52.52 & 53.06 & 52.79 & \NA & \NA & \NA & \NA \\
EANN \cite{wang2018eann} & 63.00 & 60.25 & 61.22 & 61.85 & \NA & \NA & \NA & \NA \\
SBERT-WK \cite{wang2020sbert} & 77.00 & 72.41 & 85.71 & 78.50 & \NA & \NA & \NA & \NA \\
COSMOS baseline \cite{aneja2023cosmos} & 83.25 & 86.08 & 80.67 & 83.29 & \NA & \NA & \NA & \NA \\
COSMOS on Steroids \cite{akgul2021cosmos} & 89.75 & 87.38 & 93.71 & 90.44 & \NA & \NA & \NA & \NA \\
La \textit{et al.} \cite{la2022leverage} & 89.75 & 86.72 & 94.68 & 90.53 & \NA & \NA & \NA & \NA \\
\addlinespace[3pt]
\midrule
\rowcolor{PaleBlue}
\multicolumn{9}{l}{\textbf{Evaluation set: full test set \cite{aneja2023cosmos} ($n=1{,}700$)}} \\
EANN \cite{wang2018eann} & 63.00 & \NA & \NA & \NA & \NA & \NA & \NA & \NA \\
EmbraceNet \cite{choi2019embracenet} & 68.00 & \NA & \NA & \NA & \NA & \NA & \NA & \NA \\
Jin \textit{et al.} \cite{jin2017multimodal} & 71.00 & \NA & \NA & \NA & \NA & \NA & \NA & \NA \\
COSMOS \cite{aneja2023cosmos} & 85.00 & \NA & \NA & \NA & \NA & \NA & \NA & \NA \\
Tran \textit{et al.} \cite{tran2022textual} & 89.10 & 87.80 & \secondbest{90.80} & 89.30 & \NA & \NA & [87.6, 90.5] & \NA \\
\addlinespace[2pt]
\rowcolor{PalePurple!15}
\multicolumn{9}{l}{\textbf{SEMV (ours)}} \\
\rowcolor{PalePurple!5}
\quad +Nemotron3-30B & \secondbest{91.76} & \secondbest{91.18} & \best{92.47} & \best{91.82}
& \best{0.045} & \secondbest{0.065} & [90.4, 93.0] & 0.009 \\
\rowcolor{PalePurple!5}
\quad +Gemma4-31B & 89.35 & 89.77 & 88.82 & 89.30
& 0.085 & 0.091 & [87.8, 90.7] & 0.820 \\
\rowcolor{PalePurple!5}
\quad +InternVL3.5-38B & 90.65 & 90.89 & 90.35 & 90.62
& 0.085 & 0.075 & [89.2, 91.9] & 0.140 \\
\rowcolor{PalePurple!5}
\quad +Qwen3.6-35B & \best{91.88} & \best{93.00} & 90.59 & \secondbest{91.78}
& \secondbest{0.049} & \best{0.053} & [90.5, 93.1] & 0.006 \\
\bottomrule
\end{tabularx}
\end{table}

\subsubsection{Confidence Calibration}
Predictive accuracy alone does not establish whether confidence can be used reliably for automation. We therefore evaluate \textit{calibration}, which measures how closely predicted confidence corresponds to empirical correctness. For each prediction, confidence is defined as the winning-class probability obtained from the mean subclaim score using Eq.~\ref{eq:prob_map}.

\paragraph{Metrics}
We quantify confidence calibration using two metrics, which assess confidence–accuracy agreement and probabilistic prediction error.

\textit{Expected Calibration Error (ECE)} measures the gap between reported confidence and empirical correctness. With $L$ equal-width bins ${B_\ell}{\ell=1}^{L}$ formed from ${\gamma_i}{i=1}^{n}$, $\mathsf{acc}(B_\ell)=|B_\ell|^{-1}\sum_{i\in B_\ell}\mathbf{1}[\hat y_i=y_i]$, and $\mathsf{conf}(B_\ell)=|B_\ell|^{-1}\sum_{i\in B_\ell}\gamma_i$:

\begin{equation}
\mathsf{ECE}
=
\sum_{\ell=1}^{M}
\frac{|B_\ell|}{n}
\left|
\mathsf{acc}(B_\ell)
-
\mathsf{conf}(B_\ell)
\right|.
\end{equation}

\textit{Brier Score} measures the probabilistic accuracy of the predicted OOC probabilities and penalizes confident errors more strongly, as:
\begin{align}
\mathsf{Brier}
=
\frac{1}{n}
\sum_{i=1}^{n}
(p_i-b_i)^2,
\end{align}
Lower values indicate more accurate probability estimates.

\paragraph{Results}
Table~\ref{tab:cosmos_prior_results} shows that predictive strength and confidence quality are related but distinct. Nemotron3 attains the lowest ECE (0.045), indicating the closest agreement between reported confidence and empirical correctness under the adopted calibration measure, while ranking second in Brier score (0.065). Qwen3.6 attains the lowest Brier score (0.053), indicating the most accurate probabilistic predictions, and the second-lowest ECE (0.049), so its confidence estimates are both well-aligned in aggregate and accurate per record.
InternVL3.5 remains competitive in Brier score (0.075) but shares the highest ECE (0.085) with Gemma4, which is also weakest in Brier score (0.091). Notably, Gemma4 still achieves a competitive F1 of 89.30\%, showing that strong predictive performance alone does not guarantee that confidence estimates are equally trustworthy.

Overall, SEMV performs consistently across backbones, with all variants matching or exceeding the strongest comparable baseline in accuracy and three improving F1. Qwen3.6 provides the strongest profile for confidence-aware verification, combining the highest accuracy (91.88\%) and precision (93.00\%) with the best Brier score. Nemotron3 favors OOC sensitivity, achieving the highest recall (92.47\%) and F1 (91.82\%) together with the lowest ECE. InternVL3.5 offers a balanced intermediate profile, whereas Gemma4 maintains competitive predictive performance but the least reliable confidence estimates. These results indicate that SEMV is robust across heterogeneous backbones, while backbone choice determines the trade-off between sensitivity and calibration.

\subsection{Verification Report Evaluation}\label{subsec:report}
We evaluate the quality of verification reports from both real-world and controlled perspectives. We first report the official MV2026 fact-checker assessment on the private test set, then analyze held-out validation cases to examine predictive correctness, report quality, and computational efficiency across SEMV backbones.

\subsubsection{Official MV2026 Judge Evaluation}
For the real-world verification stage of MV2026 \cite{dang20262026}, we report the results of our previously submitted A-QBAF verification framework \cite{nguyen2026contestable}, which corresponds to SEMV without the self-evolving module. Let $Q_i\in[0,110]$ denote the report-quality score for case $i$. The rubric allocates $10$ points to summary and content classification, $65$ to verified evidence, $20$ to forensic analysis, $5$ to supporting evidence and cross-checking, and $10$ to clarity and structure. The official time-adjusted score is: 
\begin{equation}
\mathsf{MVScore}_i
=
\max\!\left(0,Q_i-0.001T_i^2\right),
\end{equation}
where $T_i$ is the number of hours between case release and report submission. The aggregate challenge result follows the organizer-provided aggregation over the official cases.

Our method achieves an official aggregate judge score of $714.95$, compared with $544.11$ for Team matchalatte, $461.37$ for Team dashlab, and $441.54$ for Deep Verify \cite{nguyen2026deepverify}. SEMV obtains the highest reported score on each of the ten cases, with organizer-reported case scores ranging from $59.49$ to $83.00$. This evaluation directly measures the quality of the final verification reports under the challenge's expert assessment protocol and serves as our primary real-world measure of report performance.

\begin{table}[ht]
\centering
\caption{Official MV2026 private-test judge results over ten cases. Our submitted A-QBAF system \cite{nguyen2026contestable} corresponds to SEMV without the self-evolving module. Best in \textbf{bold}.}
\label{tab:mv2026_private_judge}
\small
\setlength{\tabcolsep}{3.7pt}
\resizebox{\textwidth}{!}{
\begin{tabular}{lrrrrrrrrrr|r}
\toprule
\textbf{Method}
& \textbf{\#01}
& \textbf{\#02}
& \textbf{\#03}
& \textbf{\#04}
& \textbf{\#05}
& \textbf{\#06}
& \textbf{\#07}
& \textbf{\#08}
& \textbf{\#09}
& \textbf{\#10}
& \textbf{Total $\uparrow$} \\
\midrule

Team matchalatte
& 63.92 & 59.42 & 59.91 & 58.42 & 55.91
& 51.91 & 49.41 & 33.91 & 50.91 & 60.40
& 544.11 \\

Team dashlab
& 54.76 & 43.26 & 41.75 & 43.74 & 43.74
& 50.23 & 47.23 & 61.22 & 37.22 & 38.22
& 461.37 \\

Deep Verify \cite{nguyen2026deepverify}
& 51.73 & 43.71 & 50.71 & 32.71 & 45.70
& 58.20 & 58.20 & 44.69 & 18.69 & 37.19
& 441.54 \\
\midrule
\rowcolor{PalePurple!10}
\textbf{SEMV (Ours)}
& \textbf{83.00}
& \textbf{77.50}
& \textbf{73.50}
& \textbf{65.49}
& \textbf{67.99}
& \textbf{65.49}
& \textbf{78.49}
& \textbf{80.49}
& \textbf{59.49}
& \textbf{63.49}
& \textbf{714.95} \\
\bottomrule
\end{tabular}
}
\end{table}

\subsubsection{Held-Out Validation Evaluation}
In addition to the official MV2026 private-test evaluation, we conduct a controlled analysis on the 10 held-out MV2026 validation cases to examine differences among SEMV backbones. We evaluate predictive performance, verification report quality, and computational efficiency.

\paragraph{Verification Metrics}
We evaluate both the correctness of the final decision and the reliability of the generated verification report. 

For predictive performance, we measure case-level \emph{accuracy} and \emph{abstention rate}. Both quantities are reported as percentages across all $N=10$ cases, with abstentions counted as errors. Let $\mathcal{A}=\{\text{\textsf{uncertain}}, \text{\textsf{insufficient-evidence}}\}$ be the set of abstention labels. Metrics are defined as:
\begin{equation}
    \mathrm{Acc}=\frac{100}{N}\sum_{i=1}^{N}\mathbf{1}\!\left[\hat y_i=y_i\right],
\qquad
\mathrm{Abst}=\frac{100}{N}\sum_{i=1}^{N}\mathbf{1}\!\left[\hat y_i\in\mathcal{A}\right].
\end{equation}

For report quality, we measure \textit{Evidence-URL F1}. Let $\Urls_i$ and $\Urls_i^{*}$ denote the cited and reference URL sets for case $i$, respectively. We compute:
\begin{equation}
\mathrm{Prec}_i=\frac{|\Urls_i\cap\Urls_i^{*}|}{|\Urls_i|},
\qquad
\mathrm{Rec}_i=\frac{|\Urls_i\cap\Urls_i^{*}|}{|\Urls_i^{*}|},
\qquad
\mathrm{F1}_i=
\frac{2|\Urls_i\cap\Urls_i^{*}|}
{|\Urls_i|+|\Urls_i^{*}|},
\label{eq:url_f1}
\end{equation}
where $\mathrm{F1}_i=0$ when no URL is cited. We report the macro-average $\mathsf{URL\text{-}F1}=100N^{-1}\sum_{i=1}^{N}\mathrm{F1}_i$. Higher values indicate more precise and complete evidence retrieval and citation.

\textit{Temporal Grounding} is evaluated from two related perspectives: whether the report correctly recovers the time of the depicted event and whether it avoids confusing source publication time with event time. Let $r_i^{t}=1$ when the event time is correctly recovered under the annotation protocol, and $c_i^{t}=1$ when a publication timestamp is incorrectly presented as the event time. We compute:
\begin{equation}
\mathsf{Temp}
=
\frac{100}{N}\sum_{i=1}^{N}r_i^{t},
\qquad
\mathsf{TimeConf}
=
\frac{100}{N}\sum_{i=1}^{N}c_i^{t}.
\end{equation}
Higher $\mathsf{Temp}$ and lower $\mathsf{TimeConf}$ indicate more reliable temporal grounding.

\textit{Required-Field Coverage} measures the completeness of the generated verification report. With $M=14$ required fields and $z_{ij}=1$ when field $j$ is completed for case $i$, we compute:
\begin{equation}
\mathsf{Cov}
=
\frac{100}{NM}
\sum_{i=1}^{N}\sum_{j=1}^{M} z_{ij}.
\end{equation}

\paragraph{Computational Efficiency Metrics}
We measure computational efficiency under matched cached tool outputs and report the median per-case values. \textit{Wall-clock time} (in minutes) measures elapsed time from the start of SEMV inference to generation of the final verification output, with external tool results replayed from the same cache across backbones. \textit{Output tokens} count all tokens generated across every LLM call within a case, including both model-thinking and final-response tokens. \textit{LLM calls} count all backbone invocations made throughout the end-to-end SEMV pipeline.

\paragraph{Results}

Table~\ref{tab:mv_validation_results} shows that Qwen3.6 provides the strongest overall verification performance. It correctly verifies $9/10$ cases with no abstentions, while achieving the highest Evidence-URL F1 ($83.5\%$), the best temporal grounding ($9/10$ correctly recovered event times), the highest required-field coverage ($97.1\%$), and no publication-time/event-time confusion. Nemotron3 is the strongest alternative, correctly verifying $8/10$ cases without abstention and achieving the second-highest Evidence-URL F1 ($82.0\%$) with $96.4\%$ field coverage. InternVL3.5 also correctly verifies $8/10$ cases and recovers event time in $8/10$ cases, but abstains once and obtains slightly lower Evidence-URL F1 ($81.1\%$) and field coverage ($95.7\%$). Gemma4 shows the weakest overall verification profile, correctly verifying $7/10$ cases and attaining the lowest Evidence-URL F1 ($79.2\%$), temporal grounding ($7/10$), and tied-lowest field coverage ($95.7\%$). These results indicate that similar case-level accuracy can still be associated with meaningful differences in evidence alignment, temporal grounding, and report completeness.

For computational efficiency, Qwen3.6 achieves the lowest median wall-clock time at $4.8$ minutes per case, closely followed by Nemotron3 at $4.9$ minutes. Nemotron3, however, produces the largest generation workload, with $12.4$k output tokens and $32$ LLM calls. Gemma4 generates the fewest output tokens ($10.6$k) and, together with InternVL3.5, requires the fewest LLM calls ($30$). InternVL3.5 has the highest wall-clock time at $9.1$ minutes despite its relatively low token count ($10.9$k) and number of calls. Overall, Qwen3.6 provides the best balance between verification quality and inference latency, whereas Gemma4 and InternVL3.5 reduce generation workload but require substantially longer execution times.

\begin{table}[t]
\centering
\caption{
Results on the 10 held-out MV2026 validation cases across predictive performance, report quality, and computational efficiency. $\uparrow$/$\downarrow$ denote metrics for which higher/lower values are better. Best in \textbf{bold}, second best \underline{underlined}.
}
\label{tab:mv_validation_results}
\small

\resizebox{\textwidth}{!}{
\begin{tabular}{lccccccccc}
\toprule
& \multicolumn{2}{c}{\textbf{Pred. Perf.}}
& \multicolumn{4}{c}{\textbf{Report Quality}}
& \multicolumn{3}{c}{\textbf{Computational Efficiency}} \\
\cmidrule(lr){2-3}
\cmidrule(lr){4-7}
\cmidrule(lr){8-10}

\textbf{SEMV}
& \textbf{Acc.$\uparrow$}
& \textbf{Abst.$\downarrow$}
& \textbf{URL F1$\uparrow$}
& \textbf{Temp.$\uparrow$}
& \textbf{Conf.$\downarrow$}
& \textbf{Cov.$\uparrow$}
& \textbf{Wall$\downarrow$}
& \textbf{Tok.$\downarrow$}
& \textbf{Calls$\downarrow$} \\
\midrule

+Nemotron3
& 8/10
& \best{0/10}
& $82.0\pm0.6$
& 8/10
& 1/10
& 96.4
& \underline{4.9}
& 12.4k
& 32 \\

+Gemma4
& 7/10
& 1/10
& $79.2\pm0.9$
& 7/10
& 1/10
& 95.7
& 8.2
& \best{10.6k}
& \best{30}\\

+InternVL3.5
& 8/10
& 1/10
& $81.1\pm0.7$
& 8/10
& 1/10
& 95.7
& 9.1
& \underline{10.9k}
& \best{30}
\\

+Qwen3.6
& \best{9/10}
& \best{0/10}
& $\bm{83.5\pm0.5}$
& \best{9/10}
& \best{0/10}
& \best{97.1}
& \best{4.8}
& 11.2k
& \underline{31}
\\

\bottomrule
\end{tabular}
}
\end{table}

\subsection{Self-Evolving Memory Evaluation}\label{subsec:memory}
This section carries the paper's principal empirical claim. We evaluate \emph{consolidation} (how candidates are validated and stored), \emph{retrieval and reuse} (whether useful records are found and actually incorporated), and \emph{transfer safety} (whether enabling memory corrects errors without harming previously correct decisions).

\subsubsection{Metrics}
\paragraph{Consolidation Metrics} We evaluate how memories generated by post-prediction reflection are verified and incorporated into long-term memory. The \textit{Grounded Pass} rate is the proportion of generated candidates that pass grounding, provenance, and consistency verification (refer Eq~\ref{eq:memory_verification}). 
Among the grounded-passed candidates, the \textit{Promotion} rate measures those stored as new long-term records. The \textit{Merge} rate measures those incorporated into existing records. The \textit{Conflict Retention} rate measures those preserved as unresolved conflicts. The \textit{Rejection} rate measures the proportion of items excluded from long-term storage.

\paragraph{Reuse Metrics}
We evaluate whether relevant memories are successfully retrieved and effectively reused during reasoning for the current verification case.
\textit{Relevance} is the proportion of retrieved records that match the current claim, evidence context, and task scope.
\textit{Successful Use} is measured over cited memory-use events and denotes the proportion that pass provenance and conflict checks and are retained after argument verification because they support a correct decision or correct an otherwise erroneous one.

\paragraph{Transfer Safety Metrics}
Transfer is evaluated only on matched memory-off and memory-on runs with the same case identifiers and evaluation settings. Positive transfer (\textit{Trans}$^{+}$) is the proportion of cases that are incorrect without memory but become correct when memory is enabled. Negative transfer (\textit{Trans}$^{-}$) is the proportion of cases that are correct without memory but become incorrect with memory.

\subsubsection{Results}

\begin{table*}[t]
\centering

\begin{minipage}[t]{0.29\textwidth}
\vspace{0pt}
\centering
\captionof{table}{Memory consolidation outcomes with Qwen3.6 over 184 reflected candidates. Indented rates are conditional on grounding pass.}
\label{tab:memory_consolidation}
\small
\setlength{\tabcolsep}{5pt}
\resizebox{\linewidth}{!}{
\begin{tabular}{@{}lc@{}}
\toprule
\textbf{Metric} & \textbf{Rate (\%)} \\
\midrule
\rowcolor{gray!10}
Grounding Pass
& 61.4 \\
$\hookrightarrow$ Promotion
& 42.5 \\
$\hookrightarrow$ Merge
& 18.6 \\
$\hookrightarrow$ Conflict Retention
& 7.1 \\
$\hookrightarrow$ Rejection
& 31.8 \\
\bottomrule
\end{tabular}
}
\end{minipage}
\hfill
\begin{minipage}[t]{0.70\textwidth}
\vspace{0pt}
\centering
\captionof{table}{Memory ablation on reuse, and transfer safety with Qwen3.6 on COSMOS test set. $\mathsf{MU}=$ memory update, $\mathsf{VC}=$ verified consolidation, and $\mathsf{CR}=$ conflict retention. Transfer rates use matched memory-on/off cases. $\uparrow$/$\downarrow$ denote metrics for which higher/lower values are better. Best in \textbf{bold}, second best \underline{underlined}.}
\label{tab:memory_results}
\small
\resizebox{\linewidth}{!}{
\begin{tabular}{@{}l|c|cc|cc@{}}
\toprule
\textbf{Memory Config.}
& \textbf{Acc. $\uparrow$}
& \textbf{Rel. $\uparrow$}
& \textbf{Succ. Use $\uparrow$}
& \textbf{Trans$^{+}\uparrow$}
& \textbf{Trans$^{-}\downarrow$} \\
\midrule

$\varnothing$
& 86.06
& \NA
& \NA
& \NA
& \NA \\

$\mathcal{M}^{\mathrm{epi}}$
& 87.29
& 82.4
& 84.1
& 11.4
& \underline{0.4} \\

$\mathcal{M}^{\mathrm{epi}}+\mathsf{MU}$
& 87.59
& 68.7
& 71.3
& \best{46.4}
& 5.7 \\

$\mathcal{M}^{\mathrm{epi}}+\mathsf{MU}+\mathsf{VC}$
& \underline{90.71}
& \underline{91.2}
& \underline{85.7}
& 36.7
& 0.5 \\

\rowcolor{PalePurple!6}$\mathcal{M}^{\mathrm{epi}}+\mathsf{MU}+\mathsf{VC}+\mathsf{CR}$
& \best{91.88}
& \best{93.6}
& \best{88.5}
& \underline{43.0}
& \best{0.2} \\

\bottomrule
\end{tabular}
}
\end{minipage}

\end{table*}

\paragraph{Consolidation is selective}
Table~\ref{tab:memory_consolidation} shows that long-term memory is formed selectively rather than by directly accumulating all reflected experiences.
Of the 184 candidate records produced by post-prediction reflection, 61.4\% pass grounding, provenance, and consistency verification. Among these grounded candidates, 42.5\% are promoted as new records, indicating that a substantial fraction contributes genuinely new reusable knowledge. 18.6\% are merged into existing memories, showing that repeated experiences can reinforce previously stored knowledge without unnecessary duplication. 7.1\% are retained explicitly as unresolved conflicts, preserving contradictory evidence instead of silently overwriting prior records. The remaining 31.8\% are rejected from long-term storage, reflecting the role of consolidation in filtering candidates that are not sufficiently suitable for future reuse.

\paragraph{Verification, not Memory Volume, is what matters}
Table~\ref{tab:memory_results} isolates the mechanism.
Episodic memory alone improves accuracy from 86.06\% to 87.29\%, with 82.4\% retrieval relevance, 84.1\% successful use, and only 0.4\% negative transfer. Allowing unrestricted memory updates yields the highest positive transfer at 46.4\%, indicating that accumulating additional experience creates more opportunities to correct baseline errors. However, relevance falls to 68.7\% and successful use to 71.3\%, while negative transfer rises sharply to 5.7\%. The resulting accuracy gain is therefore limited to 87.59\%. This demonstrates that a memory can intervene frequently and correct many errors, yet still degrade overall reliability through irrelevant or harmful reuse. 

Verified consolidation substantially changes this trade-off. Adding verified consolidation raises relevance from 68.7\% to 91.2\% and successful use from 71.3\% to 85.7\%, while reducing negative transfer from 5.7\% to 0.5\%. Accuracy consequently increases to 90.71\%. Adding conflict retention further improves relevance to 93.6\% and successful use to 88.5\%, while reducing negative transfer to 0.2\%. The full memory configuration achieves the highest accuracy of 91.88\% while retaining a high positive-transfer rate of 43.0\%. Hence, its advantage does not arise from suppressing memory intervention altogether, but from preserving beneficial transfer while preventing previously correct predictions from being overturned. The full memory policy therefore achieves the best balance between acquiring reusable experience, maintaining retrieval quality, and limiting harmful transfer.

\subsection{Contestability}\label{subsec:contest}
We evaluate whether a human correction is translated into the intended system mutation, propagated through the correct causal dependencies, and applied without unnecessarily changing already-correct state.
Fig.~\ref{fig:contestation_paths} shows the three configurations under the same reviewer \emph{edit} action. 
Unstructured feedback ($\mathsf{U}$) accepts free-form input and performs a global rerun. 
Causal contestation routing ($\mathsf{C}$) converts feedback into a typed action, identifies the earliest affected stage $u^{*}$, and reruns from that point downstream.
Scoped causal revision ($\mathsf{C{+}S}$) additionally computes the affected dependency closure $\Objs^{*}$, restricting recomputation to states that can causally depend on the contested object.

\begin{figure}[ht]
    \centering
    \includegraphics[width=.95\linewidth]{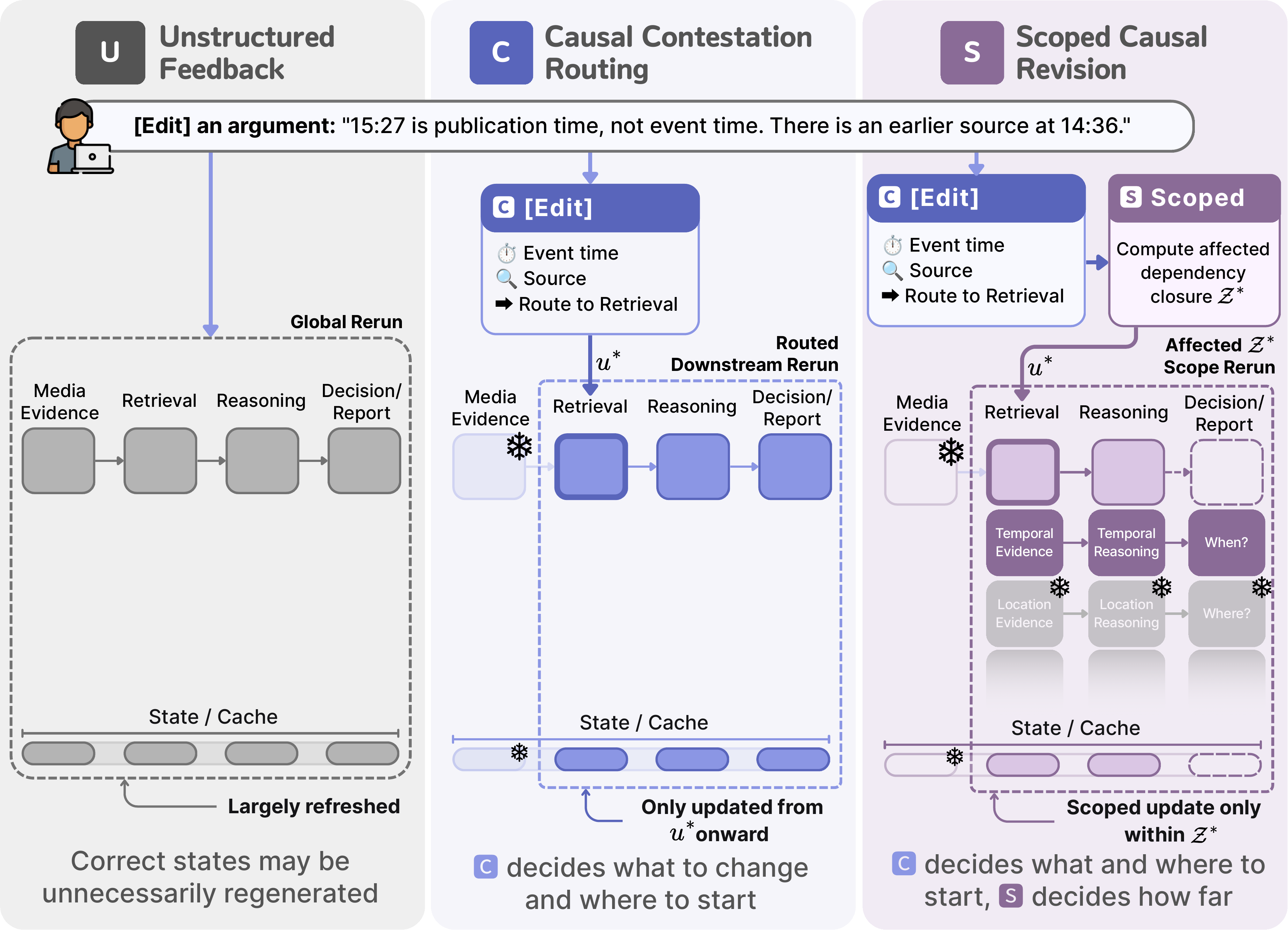}
    \caption{Contestation paths under $\mathsf{U}$, $\mathsf{C}$, and $\mathsf{C{+}S}$. Given the same human edit, $\mathsf{U}$ triggers a global rerun; $\mathsf{C}$ identifies the target and causal restart point $u^{*}$; $\mathsf{C{+}S}$ further restricts recomputation to the affected dependency closure $\Objs^{*}$, preserving unrelated cached states.}
    \label{fig:contestation_paths}
\end{figure}

\subsubsection{Quantitative Evaluation}
CTR comprises $N=50$ matched contestation episodes from 10 source cases, with the same adjudicated feedback replayed under each configuration. Of these, $N_H=40$ contain an injected report or state corruption and are used for corruption-specific process and recovery metrics. At the decision level, 30 episodes are initially erroneous and 20 initially correct. For episode $i$, let $s_i^{(0)}$, $s_i^{(H)}$, and $s_i^{*}$ denote the state before contestation, after contestation, and the adjudicated target; $\Objs_i$ the mutable state objects and $\Objs_i^{*}$ the gold dependency closure. All configurations share backbone, evidence pool, initial state, and human correction.

\paragraph{Contestation Fidelity Metrics}
To assess whether human feedback is faithfully executed through the contestation process, we propose the following metrics to measure whether the requested change is applied to the correct target, propagated through the intended causal scope, localized without altering unaffected states, and recorded in a complete revision trace.

\textit{Apply success} measures whether the requested mutation $\mu_i^{*}$ is committed to its intended target object $z_i^{*}$, while \textit{revision fidelity} additionally requires the recomputed state to match the adjudicated state within the gold scope:
\begin{equation}
\mathsf{Apply}=\frac{1}{N_H}\sum_{i=1}^{N_H}
\mathbf{1}\!\left[s^{(H)}_{i,z_i^{*}}=\mu_i^{*}\!\left(s^{(0)}_{i,z_i^{*}}\right)\right],
\qquad
\mathsf{Revision}=\frac{1}{N_H}\sum_{i=1}^{N_H}
\mathbf{1}\!\left[s_i^{(H)}\big|_{\Objs_i^{*}}=s_i^{*}\big|_{\Objs_i^{*}}\right].
\label{eq:ctr_apply_revision}
\end{equation}

\textit{Localization} applies Def.~\ref{def:locality} across the $N_H$ corrupted episodes, micro-averaged over out-of-scope objects.
Whereas \textit{revision fidelity} asks whether states \emph{inside} the scope are correct, \textit{localization} asks whether states \emph{outside} it are preserved:
\begin{equation}
\mathsf{Localization}
=
\frac{
\displaystyle
\sum_{i=1}^{N_H}
\sum_{z\in\mathcal Z_i\setminus\mathcal Z_i^*}
\mathbf{1}\!\left[
z\!\left(s_i^{(H)}\right)
=
z\!\left(s_i^{(0)}\right)
\right]
}{
\displaystyle
\sum_{i=1}^{N_H}
\left|
\mathcal Z_i\setminus\mathcal Z_i^*
\right|
}.
\end{equation}

\textit{Trace Completeness} is evaluated over all \(N\) episodes. \(\mathcal K\) denotes the seven required trace elements: the requested action, target, rationale, routing decision, affected scope, before-after object identifiers, and resulting decision delta. Let \(h_{ik}=1\) if trace element \(k\in\mathcal K\) is present in episode \(i\), and \(0\) otherwise. The micro-average trace completeness is:
\begin{equation}
\mathsf{Trace}
=
\frac{
\displaystyle
\sum_{i=1}^{N}
\sum_{k\in\mathcal K}
h_{ik}
}{
N|\mathcal K|
}.
\label{eq:ctr_trace}
\end{equation}

Because the 50 CTR episodes derive from only 10 source cases, we estimate uncertainty using a nonparametric cluster bootstrap at the source-case level. In each of 10,000 bootstrap replicates, 10 source cases are sampled with replacement together with all associated episode-level outcomes, and the metric is recomputed from these stored results. The 2.5th and 97.5th percentiles define the 95\% confidence interval.

\paragraph{Contestation Effectiveness Metrics}
To assess whether contestation improves the resulting verification output, we measure whether corrupted report content is recovered, erroneous decisions are corrected, initially correct decisions are preserved, and decision changes correspond to genuine corrections.

\textit{Correction} measures the fraction of initially erroneous decisions that become correct after contestation, while \textit{Induced Error} captures the opposite failure mode by measuring how often contestation damages an initially correct decision, as:
\begin{equation}
\mathsf{Correction}
=
\frac{
\displaystyle
\sum_{i=1}^{N}
\mathbf{1}\!\left[
\hat y_i^{(0)}\neq y_i
\land
\hat y_i^{(H)}=y_i
\right]
}{
\displaystyle
\sum_{i=1}^{N}
\mathbf{1}\!\left[
\hat y_i^{(0)}\neq y_i
\right]
},
\quad
\mathsf{InducedError}
=
\frac{
\displaystyle
\sum_{i=1}^{N}
\mathbf{1}\!\left[
\hat y_i^{(0)}=y_i
\land
\hat y_i^{(H)}\neq y_i
\right]
}{
\displaystyle
\sum_{i=1}^{N}
\mathbf{1}\!\left[
\hat y_i^{(0)}=y_i
\right]
}.
\end{equation}

\textit{Decision-Change Precision} (\(\mathsf{DCP}\)) evaluates whether decision changes are justified by successful error correction rather than unnecessary or harmful label changes. It is the proportion of all changed decisions that correct an initially erroneous prediction:
\begin{equation}
\mathsf{DCP}
=
\frac{
\displaystyle
\sum_{i=1}^{N}
\mathbf{1}\!\left[
\hat y_i^{(0)}\neq y_i
\land
\hat y_i^{(H)}=y_i
\right]
}{
\displaystyle
\sum_{i=1}^{N}
\mathbf{1}\!\left[
\hat y_i^{(H)}
\neq
\hat y_i^{(0)}
\right]
}.
\label{eq:decision_change_precision}
\end{equation}

\textit{Report Recovery} measures whether corrupted report fields are restored to their adjudicated values after contestation. Let \(\mathcal G_i\) denote the set of corrupted report fields in episode \(i\), \(r_{if}^{(H)}\) the revised value of field \(f\), and \(r_{if}^*\) its adjudicated accepted value. Report Recovery is:
\begin{equation}
\mathsf{RepRec}
=
\frac{
\displaystyle
\sum_{i=1}^{N_H}
\sum_{f\in\mathcal G_i}
\mathbf{1}\!\left[
r_{if}^{(H)}
\equiv
r_{if}^*
\right]
}{
\displaystyle
\sum_{i=1}^{N_H} |\mathcal G_i|
},
\label{eq:report_recovery}
\end{equation}
where \(\equiv\) denotes adjudicated semantic equivalence rather than exact string identity.

\paragraph{Computational Cost Metrics}
Beyond producing correct revisions, an effective contestation mechanism should avoid unnecessary downstream recomputation and efficiently return the revised result. We therefore measure system turnaround after feedback submission, the fraction of downstream transitions avoided, and the resulting computation saved relative to global recomputation.

\textit{Revision Time} measures the system turnaround after human feedback is submitted. For configuration \(q\), the time is measured from intervention submission until the revised state and report become available and is reported as the median across all \(N\) matched episodes. Human inspection and feedback-entry time are excluded to isolate the post-submission revision cost.

\textit{Recomputation Avoided} measures the proportion of downstream state transitions skipped relative to the global-rerun configuration \(\mathsf{U}\). For each episode, we compare the transitions re-executed under configuration \(q\) with those required by global recomputation and report the median percentage avoided over all matched episodes.

\textit{Compute Saved} measures the reduction in model-execution cost relative to global recomputation. For each episode, the compute cost under configuration \(q\) is compared with that under \(\mathsf{U}\), using the same backbone and hardware, and the median percentage saved is reported over all matched episodes.

\begin{figure*}[t]
\centering

\begin{minipage}[t]{0.45\textwidth}
\vspace{0pt}
\centering
\captionof{figure}{
Process-fidelity results with Qwen 3.6.
}
\includegraphics[width=\linewidth]{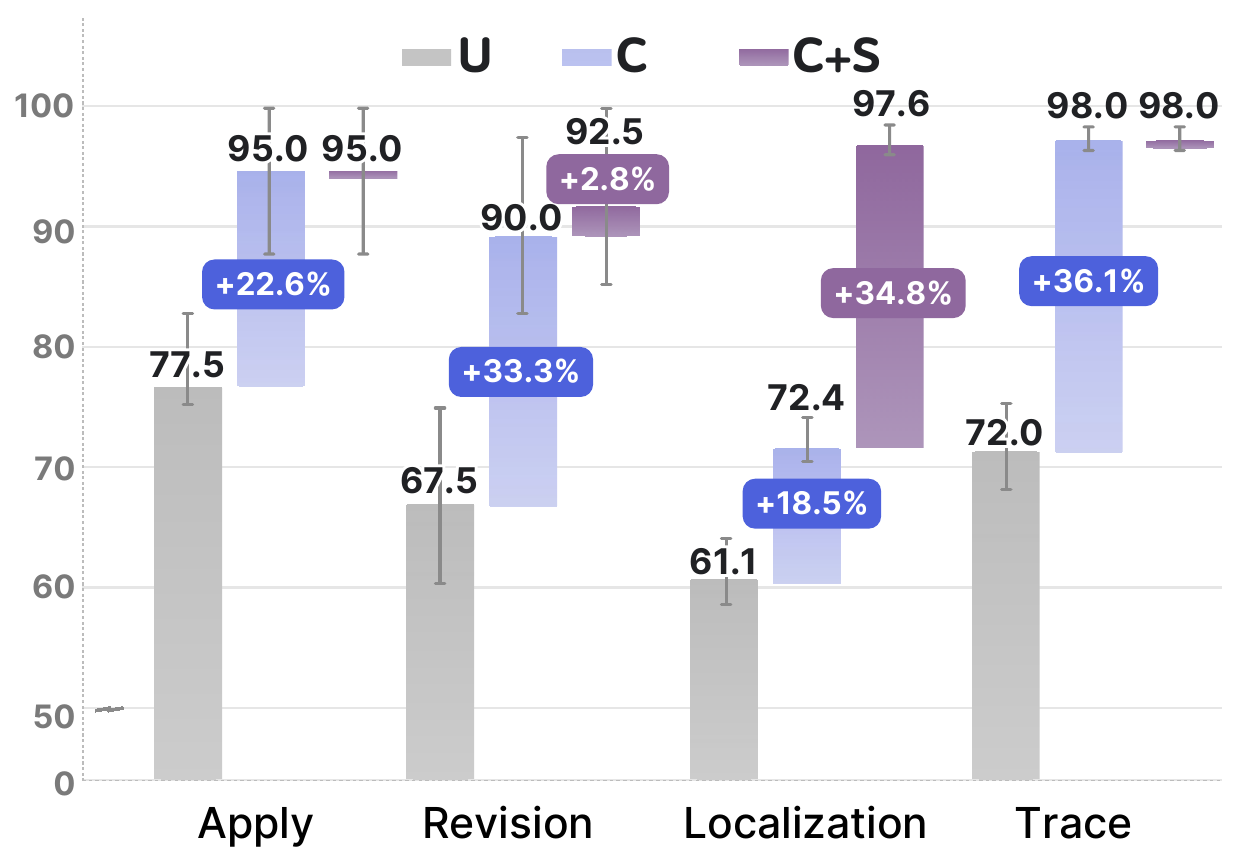}

\label{fig:ctr_fidelity}
\end{minipage}
\hfill
\begin{minipage}[t]{0.54\textwidth}
\vspace{0pt}
\centering

\captionof{table}{
Contestation effectiveness and computational cost ablation with Qwen 3.6. $\uparrow$/$\downarrow$ denote metrics for which higher/lower values are better. Best in \textbf{bold}.
}
\label{tab:ctr_repair_results}

\small
\setlength{\tabcolsep}{6pt}
\renewcommand{\arraystretch}{1.12}
\resizebox{\linewidth}{!}{
\begin{tabular}{lrr>{\columncolor{PalePurple!6}}r>{\columncolor{PalePurpleRed!15}}r}
\toprule
\textbf{Metric} & & $\mathsf{U}$ & $\mathsf{C}$ & $\mathsf{C{+}S}$ \\
\midrule
\multicolumn{5}{l}{\textit{Contestation effectiveness}} \\
\quad Correction$\uparrow$ & $/30$ & 25 & 28 & \best{29} \\
\quad Induced error$\downarrow$ & $/20$ & 2 & 1 & \best{0} \\
\quad DCP (\%)$\uparrow$ & & 92.6 & 96.6 & \best{100.0} \\
\quad Report recovery (\%)$\uparrow$ & & 70.0 & \best{90.0} & \best{90.0} \\
\midrule
\multicolumn{5}{l}{\textit{Computational cost}} \\
\quad Revision time (min)$\downarrow$ & & 4.80 & 3.76 & \best{2.27} \\
\quad Recomputation avoided (\%)$\uparrow$ & & 0.0 & 30.8 & \best{68.4} \\
\quad Compute saved (\%)$\uparrow$ & & 0.0 & 21.7 & \best{52.8} \\
\bottomrule

\end{tabular}
}
\end{minipage}

\end{figure*}

\paragraph{Results}
Relative to unstructured feedback, $\mathsf{C}$ raises Apply from $77.5\%$ to $95.0\%$ ($31/40$ to $38/40$), Revision from $67.5\%$ to $90.0\%$ ($27/40$ to $36/40$), and Trace from $72.0\%$ to $98.0\%$.
Apply shifts from $[75.0,82.5]$ under $\mathsf{U}$ to $[87.5,100]$ under both structured variants, Revision from $[60.0,75.0]$ to $[82.5,97.5]$ under $\mathsf{C}$ and $[85.0,100]$ under $\mathsf{C+S}$, and Trace from $[68.6,75.4]$ to $[97.1,98.9]$.
Structuring the intervention and selecting a causal restart point therefore produces a measurable and statistically supported improvement in whether the reviewer's intent is actually executed. The gain in Localization is smaller ($61.1\%$ to $72.4\%$) because everything downstream of $u^{*}$ remains eligible for regeneration.

The same pattern appears at the decision and report levels. Correction increases from 83.3\% under \(\mathsf{U}\) to 93.3\% with \(\mathsf{C}\) and 96.7\% with \(\mathsf{C+S}\), while Induced Error decreases from 10.0\% to 5.0\% and then 0.0\%. Accordingly, DCP reaches 100.0\% under \(\mathsf{C+S}\), meaning that every observed decision change corrects an initially erroneous decision. Report Recovery increases from 70.0\% to 90.0\% when causal routing is introduced and remains at 90.0\% after adding scoping. Thus, the observed report-recovery gain is associated primarily with causal routing, whereas scoped revision contributes more strongly to state preservation and computational efficiency.

The reduction in recomputation is also reflected in revision cost. Compared with global rerunning, \(\mathsf{C}\) avoids 30.8\% of downstream recomputation and saves 21.7\% of execution compute, reducing median Revision Time from 4.80 to 3.76 minutes. With \(\mathsf{C+S}\), Recomputation Avoided increases to 68.4\% and Compute Saved to 52.8\%, while median Revision Time decreases further to 2.27 minutes. Together with the increase in Localization from 72.4\% to 97.6\%, these results indicate that scoped causal revision improves efficiency by reusing unaffected state while maintaining, and slightly improving, revision and decision-repair performance.

\subsubsection{Qualitative Trace Evaluation}
In this section, we qualitatively examine how human contestation and self-evolving memory improve multimedia verification beyond the final verdict. The analysis compares the verification state of a specific CTR episode, before and after contestation, focusing on whether challenged claims by reviewers are correctly routed, revalidated, and revised without unnecessarily altering well-supported evidence. We further examine whether validated corrections can be consolidated into a reusable verification experience that improves the reasoning procedure of future cases without introducing case-specific factual knowledge.

\begin{tcolorbox}[
  enhanced,
  colback=gray!2,
  colframe=gray!40,
  boxrule=0.4pt,
  arc=1pt,
  left=4pt,
  right=4pt,
  top=4pt,
  bottom=4pt
]
\textbf{Media Input.}
{\small
The case input comprises two videos with representative key frames:
}

\medskip

\noindent
{\footnotesize [Video 1] Duration \(36.0\,\mathrm{s}\), resolution \(464\times848\)}

\smallskip

\noindent
\mediaframe{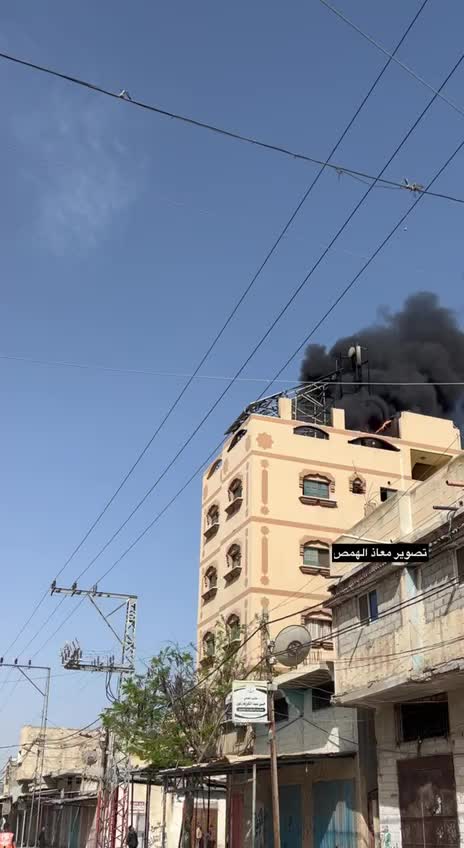}{4\,\mathrm{s}}
\keyframegap
\mediaframe{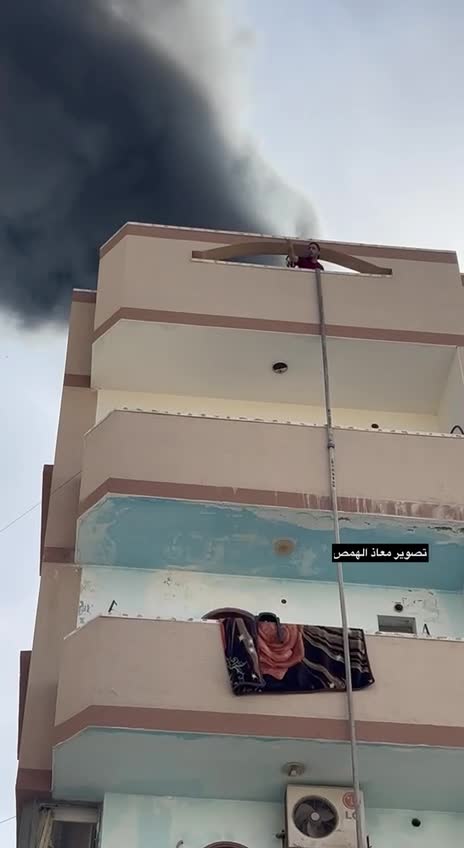}{15\,\mathrm{s}}
\keyframegap
\mediaframe{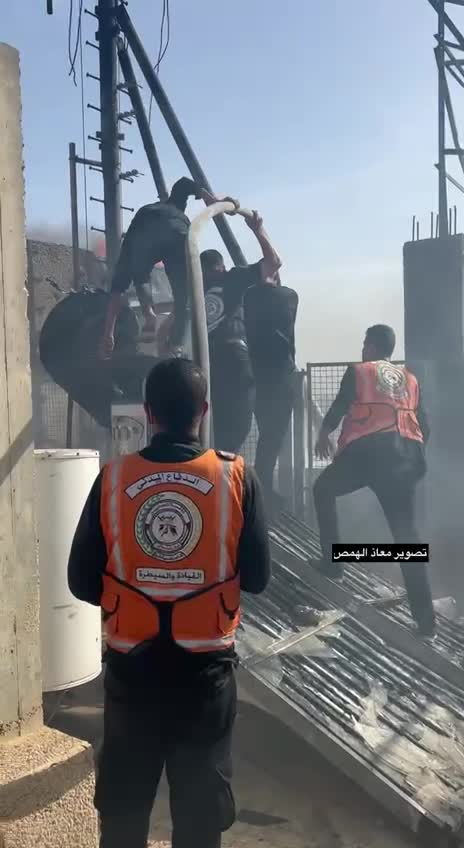}{29\,\mathrm{s}}

\medskip

\noindent
{\footnotesize [Video 2] Duration \(26.5\,\mathrm{s}\), resolution \(368\times400\)}

\smallskip

\noindent
\mediaframe{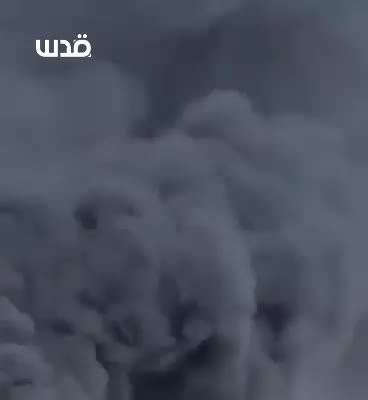}{3\,\mathrm{s}}
\keyframegap
\mediaframe{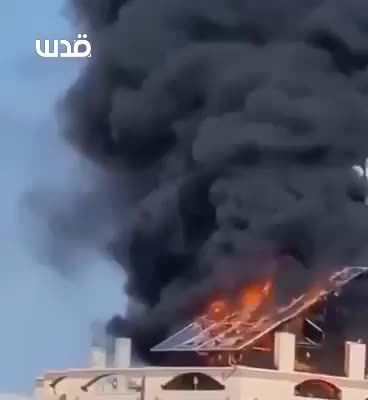}{12\,\mathrm{s}}
\keyframegap
\mediaframe{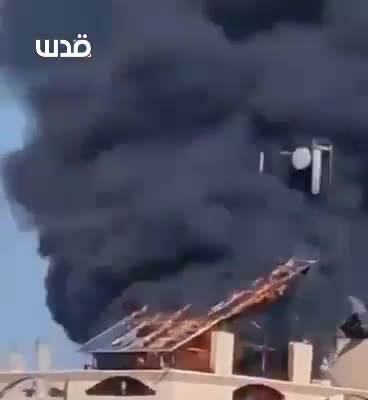}{22\,\mathrm{s}}

\end{tcolorbox}

\begin{pairedreport}[
Initial Report: Before Human Contestation and Memory Reuse
]{SEMVNavy}

\textbf{Case Summary.}
The videos show black smoke and fire at a light-yellow building in central
Rafah. The candidate coordinates are $31.27805,34.251062$, and the visual
content is broadly consistent with this location.

\textbf{Evidence and Provenance.}
The longer Telegram video appears on Moaz Al-Hams's channel at 15:27 on May 10 and provides the clearest view of the building and surrounding street. A shorter version appears to have circulated earlier on X during the same afternoon, but the first pass does not establish whether it is the earliest located publication. Both videos appear to depict the same fire.

\textbf{Where?}
The coordinates supplied by an external geolocation account are plausible.
The light-yellow facade, rooftop structure, decorative vertical lines, and
nearby power lines are broadly consistent with the proposed location in
Rafah.

\textbf{When?}
The fire probably occurred on May 10 before the longer video was published
at 15:27. The exact ignition time cannot be recovered because the available
social-media files contain no useful original timestamp metadata.

\textbf{Who and Why?}
The longer video is most likely associated with photographer Moaz Al-Hams.
An orange reflective vest visible during the response appears consistent
with Civil Defense markings. The fire appears concentrated around a rooftop
structure that some social-media posts describe as communications equipment,
although the ignition cause and responsible actor are not established.

\textbf{Verdict.}
The occurrence of the fire and its approximate location in central Rafah
are supported. The timing is provisionally placed on May 10 before 15:27,
while the cause remains unresolved.

\end{pairedreport}

\begin{humancontestbox}

\haction{REJECT}
\textbf{When (existing temporal argument $a_t^{(0)}$):}\enspace\textit{``The May 10 publication timestamp does not support the claim that the fire occurred on May 10. It only establishes a latest-known bound for the depicted event.''}

\smallskip

\haction{ADD}
\textbf{When (temporal support $a_t^{+}$; E2, E4):}\enspace \textit{``The building is still standing in the March 9 reference footage, and the earliest located fire publication is May 10 at 14:36. The event onset should therefore be bounded to after March 9 and no later than May 10 at 14:36.''}

\smallskip

\haction{EDIT}
\textbf{Where (geolocation support $a_g^{+}$; E3, E4):}\enspace\textit{``Treat the external coordinates as a candidate, not as independent proof. Revalidate them by matching the rooftop structure, vertical facade details, and power lines across the fire footage, March 9 reference footage, and satellite view.''}

\smallskip

\haction{ADD}
\textbf{Why (causal attack $a_c^{-}$):}\enspace
\textit{``The rooftop structure, black smoke, and Civil Defense responder vest do not establish the function of the rooftop equipment, the ignition cause, or the responsible actor.''}

\end{humancontestbox}

\begin{pairedreport}[
Final Report: After Human Contestation with Prior Consolidated Memory
]{SEMVTeal}

\noindent
\begin{tcolorbox}[
  enhanced,
  colback=white,
  colframe=gray!45,
  boxrule=0.4pt,
  arc=1pt,
  left=3pt,
  right=3pt,
  top=2pt,
  bottom=2pt
]
\centering
\scriptsize\sffamily
\colorbox{PaleBlue}{
  \textcolor{SEMVNavy}{\bfseries [C] Evidence Revalidation}}
\hspace{0.6em}
\colorbox{DraftRed!6}{
  \textcolor{DraftRed}{\bfseries [E] Human Contestation}}
\hspace{0.6em}
\colorbox{PaleGreen}{
  \textcolor{SEMVTeal}{\bfseries [F$\rightarrow$C] Memory Reuse}}
\end{tcolorbox}

\textbf{Case Summary.}
The videos show black smoke and fire at a light-yellow building in central Rafah. The occurrence and location can be verified, whereas the precise event onset and cause remain unresolved.

\textbf{Evidence and Provenance.}

\retrievalblock{
The source publication times are explicitly separated from event time. Quds News Network posted the shorter video at 14:36 on May 10 (E2), whereas the longer and more detailed video appears on Moaz Al-Hams's Telegram channel at 15:27 (E1). The latter is retained for visual and geolocation analysis because it contains more reference points, but it is no longer treated as the earliest located publication.
}

\textbf{Where?}

\retrievalblock{
The coordinates $31.27805,34.251062$ from the external geolocation post (E3) are treated as a candidate and independently checked. The rooftop structure, vertical stucco lines and circles, facade color, and power lines are matched across the fire video, satellite imagery, and the March 9 Al Jazeera street footage (E4).
}

\contestblock{
The burning building is distinguished from the nearby Al-Masri (Egyptian) Tower. The March 9 tower footage is used only as a reference viewpoint from which the same yellow building and surrounding landmarks can be identified, rather than as direct evidence about the later fire.
}

\textbf{Verification Method.}

\memoryblock{
A previously consolidated temporal-bounding rule from the frozen memory snapshot is made available to Module C: when original recording metadata is unavailable, a publication timestamp supplies an upper bound rather than the event time itself. In the replay, this changes the temporal verification plan from assigning a date to constructing an interval, prompting separate checks for a reliable pre-event observation and the earliest located fire publication. The memory changes the procedure only; it supplies no factual evidence about ID408.
}

\textbf{When?}

\retrievalblock{
The building is visible standing in the March 9 reference footage (E4), and the earliest located fire video was posted at 14:36 on May 10 (E2). Reverse-image searches found no earlier publication of either fire video, while the available social-media files contain no usable original recording timestamp metadata.
}

\contestblock{
The onset of the depicted fire is therefore bounded to after the March 9 reference footage and no later than May 10 at 14:36. The first-pass inference that the fire probably occurred on May 10 is withdrawn because a publication timestamp does not establish that the event began on the same date.
}

\textbf{Who?}
The longer video is most likely associated with photographer Moaz Al-Hams (E1), and an orange reflective vest visible during the response is consistent with Civil Defense markings after reverse-image comparison.

\textbf{Why?}

\contestblock{
The rooftop structure and responder presence are not treated as causal evidence. Descriptions of the rooftop structure as communications equipment are not sufficient to establish its function, an ignition mechanism, a deliberate target, or a responsible actor. The ignition cause and responsible actor therefore remain unresolved.
}

\textbf{Verdict.}

\contestblock{
The occurrence of the fire and its location at the specified coordinates in Rafah are verified. The event onset is bounded rather than assigned an unsupported exact date or time, while the ignition cause and responsible actor remain unresolved.
}

\medskip
\textbf{Sources.}
\begin{enumerate}[
  leftmargin=2.7em,
  label={\textbf{[E\arabic*]}},
  itemsep=1pt,
  topsep=2pt
]
\item Moaz Al-Hams, Telegram post:
\url{https://t.me/moathalhams/68406}

\item Quds News Network, X post:
\url{https://x.com/QudsNen/status/1788910996481814941}

\item GeoConfirmed-linked geolocation post:
\url{https://x.com/NemoAnno/status/1788950384066621534}

\item Al Jazeera, March 9 reference footage:
\url{https://www.facebook.com/watch/?v=371890269021343}
\end{enumerate}

\end{pairedreport}

\paragraph{Results}
The comparison shows that contestation primarily improves the provenance and defensibility of the reasoning rather than changing the overall verdict.

Despite reaching a broadly correct conclusion, the initial report has three weaknesses. It conflates a social-media video's publication time with the event time, relies too heavily on externally supplied coordinates, and risks extending observable cues such as rooftop structure, smoke, and responder markings into unsupported claims about equipment function, ignition mechanism, or responsibility.
Human contestation targets these weaknesses without discarding the rest of the report. The temporal challenge redirects verification from assigning a single event date to establishing a defensible interval. The geolocation challenge treats the supplied coordinates only as candidates and revalidates them through cross-source correspondences involving facade geometry, rooftop structure, power lines, satellite imagery, and earlier street-level footage. The causal challenge removes unsupported explanations while preserving the underlying visual observations. The revised report consequently treats the \texttt{14:36} publication as the earliest located post-event evidence while retaining the \texttt{15:27} video for its richer visual content, independently corroborates the location, and bounds the event between earlier footage showing the building still standing and the earliest located fire publication. The exact onset, ignition cause, and responsible actor remain unresolved where evidence is insufficient.
Scoped revision also preserves evidence and conclusions that were not challenged, showing that contestability performs \textit{targeted repair} rather than unrestricted regeneration. Prior consolidated memory contributes a reusable rule that publication timestamps provide temporal bounds when original recording metadata is unavailable, guiding separate searches for pre- and post-event evidence without supplying case-specific facts. The corrected reasoning can then be consolidated as future verification experience.

Overall, the case shows how human contestation, causal routing, scoped revision, and memory reuse jointly repair the evidence-argument-decision chain, yielding a report that is better grounded, appropriately uncertain, and easier to audit.

\section{Discussion}\label{sec:discussion}
\subsection{Main Findings}
\paragraph{Argument-centered verification improves more than the final label}
Across the full COSMOS test set, all four SEMV backbones match or exceed the strongest directly comparable baseline in accuracy, while the strongest variants also improve F1 and provide useful confidence estimates (Sec.~\ref{subsec:predictive}). The held-out MV2026 analysis further shows that models with similar case-level accuracy can differ in evidence alignment, temporal grounding, and report completeness (Sec.~\ref{subsec:report}). These results support the central design choice of treating provenance-bearing arguments as operational system state rather than as post-hoc explanations, because the same representation supports prediction, confidence estimation, report construction, audit, and later revision.

\paragraph{Verified memory is more important than memory volume}
The memory ablation isolates the main self-evolution result. Unrestricted memory update produces the highest positive transfer (46.4\%) but also increases negative transfer to 5.7\%, leaving accuracy at 87.59\%. Verified consolidation raises retrieval relevance and successful use while reducing negative transfer to 0.5\%, and retaining unresolved conflicts reduces it further to 0.2\% while reaching 91.88\% accuracy (Table~\ref{tab:memory_results}). Thus, self-evolution is better viewed as a constrained knowledge-selection problem than as continual accumulation: useful experience must be grounded, transferable, and prevented from silently overriding contradictory evidence. This empirically supports the concern that unverified agent memory can amplify locally plausible but non-transferable experience \cite{xiong2026memory,zhang2026useful}.

\subsection{Limitations}
The evaluation has two principal scope limitations. First, COSMOS is a controlled binary image--caption consistency benchmark, whereas open-world verification additionally requires source discovery, temporal and geographic reasoning, and heterogeneous media analysis. MV2026 exercises these capabilities but contains only 10 held-out validation cases and 10 private-test cases. CTR contains 50 matched episodes but only 10 independent source cases; the case-clustered bootstrap accounts for within-case dependence, yet the small number of source cases still limits population-level inference. Moreover, CTR uses adjudicated corruptions and replayed contestations rather than longitudinal interactions with independent professional fact-checkers. The official MV2026 score also evaluates the submitted A-QBAF predecessor rather than the complete self-evolving SEMV system.

Second, the full COSMOS conditions are single-run evaluations, with three-seed stability measured only on a 500-case subset. Paired experiments replay a fixed retrieval cache, which improves cross-condition comparability but does not capture the volatility of live search engines, changing webpages, or regional source availability. The sensitivity study varies one factor group at a time in a local neighborhood, and the present results cover four VLM backbones on one hardware platform. Broader multilingual benchmarks, live-retrieval replication, adversarial memory-poisoning tests, and larger human-contestation studies are therefore required before drawing deployment-level conclusions.

\section{Conclusion}
SEMV frames multimedia verification as an auditable evidence--argument--decision process that supports human revision and verified self-evolution. It combines provenance-constrained A-QBAF reasoning, causal contestation, scoped revision, and verification-gated memory consolidation. Results on COSMOS, CTR, and MV2026 show competitive verification performance, safer knowledge reuse, more faithful corrections, and lower recomputation. Overall, the findings suggest that self-evolving verification is most effective when adaptation remains traceable, contestable, and separate from factual evidence. Future work should evaluate SEMV on broader multilingual settings, live retrieval, larger human-contestation studies, and stronger attacks on memory and provenance.

\paragraph{Ethical Consideration}
Automated multimedia verification is inherently dual-use: false \emph{false-context} decisions may suppress accurate reporting, while false \emph{verified} decisions may legitimize misinformation. SEMV therefore uses abstention, human contestation, and provenance-constrained memory, and is intended for decision support rather than autonomous moderation. Low-confidence cases can be escalated through \emph{uncertain} or \emph{insufficient-evidence} outcomes (Sub.~\ref{subsec:predictive}).
Each argument must cite resolvable evidence (Def.~\ref{def:admissibility}), while consolidated memory is kept separate from evidence and cannot independently establish factual claims. The verification gate in Eq.~\ref{eq:memory_verification} further mitigates, but does not eliminate, memory-poisoning risk. As MV2026 and CTR may contain identifiable individuals, we report only verification-level outcomes and release no derived identity information. Performance may also vary across languages, regions, and source ecosystems due to uneven retrieval coverage.

\section*{Acknowledgments}
\paragraph{Funding}
This work is supported by NSERC Discovery Grant No RGPIN-2025-04478 and NSERC Discovery Supplement Award No DGECR-2025-00129.

\paragraph{Competing interests}
The author(s) declare(s) that there is no conflict of interest regarding the publication of this article.

\paragraph{Author contributions}
T. T. H. Nguyen conceived the study, designed the SEMV framework, led the implementation and experimental design, analyzed the results, and drafted the manuscript. V. T. K. Nguyen contributed to system design, implementation, experimental evaluation, and manuscript revision. H.-L. Cao contributed to the argumentation methodology, experimental design, and manuscript revision. P. Ho contributed to argumentation methodology and manuscript revision. T. T. Nguyen contributed to data processing, implementation, and manuscript revision. V. Pham contributed to the implementation. H. Cao supervised the research, acquired funding, and revised the manuscript. All authors reviewed and approved the final manuscript.

\section*{Data Availability}
COSMOS \cite{aneja2023cosmos} is available through the official project at \url{https://github.com/shivangi-aneja/COSMOS}, subject to its research-use access agreement; the original data cannot be redistributed by the authors. MV2026 \cite{dang20262026} is distributed by the organizers through the official challenge website at \url{https://sites.google.com/view/mv2026}, where access is administered by the organizers. Accordingly, the original MV2026 multimedia cases are not redistributed with this article. Our reviewer-based CTR dataset is available upon reasonable request, excluding the underlying MV2026 multimedia content. 

\section*{Supplementary Materials}
Materials and Methods: notation and formal definitions; functional roles of SEMV components; software and tool configuration; configuration constants; three-run stability protocol; and configuration-sensitivity protocol.

\noindent Fig.~\ref{fig:configuration_sensitivity}: Sensitivity to configuration constants on the COSMOS validation sensitivity subset.

\noindent Tables~\ref{tab:notation}--\ref{tab:sensitivity}: principal SEMV notation; functional roles; tool inventory and software versions; configuration constants; three-run stability evaluation; and sensitivity to configuration constants.

\printbibliography

\newpage
\appendix
\setcounter{table}{0}
\renewcommand{\thetable}{S\arabic{table}}

\setcounter{figure}{0}
\renewcommand{\thefigure}{S\arabic{figure}}

\setcounter{equation}{0}
\renewcommand{\theequation}{S\arabic{equation}}

\section*{Supplementary Materials}

\setcounter{subsection}{0}
\renewcommand{\thesubsection}{A\arabic{subsection}}

\subsection{Notation}
\label{app:notation}

Table~\ref{tab:notation} summarizes the principal notation used in the SEMV formulation. Symbols are grouped by their roles in case representation and output, evidence and arguments, multi-agent state, A-QBAF reasoning, human contestation, and verified memory consolidation.

\begin{table}[h]
\centering
\caption{Principal notation used in SEMV, grouped by case and output, evidence and arguments, multi-agent state, A-QBAF reasoning, human contestation, and verified memory consolidation.}
\label{tab:notation}
\footnotesize
\setlength{\tabcolsep}{5pt}
\renewcommand{\arraystretch}{1.02}
\resizebox{\textwidth}{!}{
\begin{tabular}{@{}lll@{}}
\toprule
\textbf{Symbol} & \textbf{Type} & \textbf{Meaning} \\
\midrule

\rowcolor{gray!6}
\multicolumn{3}{@{}l}{\textit{Case and output}}\\
$x$, $x_i$
& instance
& verification instance; $i$ indexes cases \\

$\mathcal{V}_x=\{m_j\}_{j=1}^{n_x}$
& set
& input images/videos of case $x$ \\

$c$
& text
& main textual claim \\

$c_k$, $K$
& subclaim, int
& scoped subclaim $k=1,\dots,K$; number of subclaims \\

$z$
& record
& optional contextual metadata \\

$\delta_{\mathrm{data}}$
& adapter
& dataset adapter and output contract \\

$\hat y$, $\mathbf{p}$, $\gamma$
& label, dist., scalar
& predicted label, class distribution, confidence \\

$\mathcal{J}$, $\mathcal{T}$, $\xi$
& records
& justification, auditable trace, output uncertainty codes \\

\midrule
\rowcolor{gray!6}
\multicolumn{3}{@{}l}{\textit{Evidence and arguments}}\\

$\mathcal{E}_t$, $e$
& set, item
& evidence at logical step $t$; an evidence item \\

$r_e$, $q_e$, $\rho_e$, $\lambda_e$, $\upsilon_e$
& fields
& reliability, relevance, provenance, location, uncertainty \\

$\mathcal{G}_t$, $a$
& set, item
& arguments in state; a single argument \\

$\sigma_a$, $t_a$, $w_a$, $v_a$, $\upsilon_a$
& fields
& stance, text, strength, verifier outcome, uncertainty \\

$E(a)$, $M(a)$
& id sets
& cited evidence identifiers; cited memory identifiers \\

\midrule
\rowcolor{gray!6}
\multicolumn{3}{@{}l}{\textit{Agents and state}}\\

$\mathcal{A}=\{A_1,\dots,A_n\}$
& set
& functional agents; $A_i$ distinguishes agents from argument $a$ \\

$\mathcal{S}$, $s$, $s_t$
& space, state
& case-state space; a committed state \\

$\mathcal{U}_i$, $u$
& set, action
& action set of agent $i$; an agent action \\

$\mathcal{L}$, $\mathcal{P}$, $T$, $\mathcal{I}$
& set, graph, function, set
& messages, dependency graph, transition, invariants \\

$\Omega_i$, $\pi_i$
& functions
& local observation and agent policy \\

\midrule
\rowcolor{gray!6}
\multicolumn{3}{@{}l}{\textit{Argumentation}}\\

$Q_k=\langle V_k,R_k^{+},R_k^{-},\beta_k\rangle$
& QBAF
& nodes, support/attack edges, and base-score function \\

$S_k$, $A_k$, $\Delta_k$
& scalars
& total support, total attack, net argumentative effect \\

$h(\cdot)$
& function
& saturating influence function (Eq.~\ref{eq:qbaf_propagation}) \\

$s_k$, $\bar{s}$, $d_k$
& scalars, label
& subclaim score, mean subclaim score, subclaim decision \\

\midrule
\rowcolor{gray!6}
\multicolumn{3}{@{}l}{\textit{Contestation}}\\

$h$, $H$
& action, batch
& human review action; non-empty review batch \\

$\mathcal{Z}$, $\mathcal{Z}_H$
& sets
& observable state objects; intended dependency scope of $H$ \\

$u_H$
& stage
& earliest stage required by review batch $H$ \\

$\mathsf{Loc}(H)$
& scalar
& contestation locality (Eq.~\ref{eq:contestation_locality}) \\

\midrule
\rowcolor{gray!6}
\multicolumn{3}{@{}l}{\textit{Memory}}\\

$m$, $\tau$
& set, record, type
& memory record; memory type \\

$u$, $c_u$
& candidate, scalar
& reflection candidate and its confidence \\

$E_u$, $A_u$
& id sets
& grounding evidence and argument identifiers \\

$\operatorname{Ind}(C)$
& set
& independent subset selected from cluster $C$ \\

$\alpha_m$, $\beta_m$, $c(m)$, $r(m)$
& scalars
& lifecycle parameters, record confidence, and conflict ratio \\

\bottomrule
\end{tabular}}
\end{table}

\newpage
\subsection{Functional Roles}

\begin{table}[h]
\centering
\caption{Functional roles, their local observations, principal actions, and primary outputs or state updates.}
\label{tab:agents}
\footnotesize
\rowcolors{2}{white}{gray!8}
\begin{tabularx}{\textwidth}{
>{\raggedright\arraybackslash}p{2.2cm}
>{\raggedright\arraybackslash}p{3.1cm}
>{\raggedright\arraybackslash}p{5.1cm}
>{\raggedright\arraybackslash}X}
\toprule
\textbf{Agent role} &
\textbf{Local observation $\Omega_i(s)$} &
\textbf{Principal action} &
\textbf{Primary output / state update} \\
\midrule

Canonicalization &
Native case and adapter contract &
Validate and normalize the dataset instance; enforce pre-prediction leakage checks &
$B$ including the media manifest; adapter trace \\

Perception &
Canonical media, main claim, and context &
Extract metadata, scene-aware keyframes, OCR, ASR, visual descriptions,
forensic signals, and local reverse matches &
Initial media evidence $\mathcal{E}^{\mathrm{media}}$ \\

Claim decomposition &
$B$ and extracted media evidence &
Produce scoped what, where, when, who, why, and authenticity subclaims &
$\mathcal{C}$ including initial query seeds \\

Memory retrieval &
Subclaim, case context, media evidence, and active long-term memory &
Retrieve compatible active long-term records as planning guidance &
Retrieved view of $\mathcal{M}$ and usage trace \\

Research planning &
Subclaim, current evidence, and retrieved memory &
Select questions, queries, source types, and uncertainty checks &
$\mathcal{R}$ \\

Evidence retrieval &
Research plan, search adapters, and cached evidence &
Search supplied or cached evidence, web, news, fact-check, geolocation,
and online reverse-image sources &
Candidate external evidence $\mathcal{E}^{\mathrm{ext}}_k$ \\

Evidence validation &
Candidate evidence and provenance &
Complete provenance, merge duplicate identifiers, rank evidence,
and construct evidence links &
Normalized $\mathcal{E}$ and $G^E$ \\

Proponent / Opponent &
One subclaim and its selected evidence &
Construct evidence-linked support and attack arguments &
Candidate $\mathcal{G}_k$ \\

Argument verification &
Argument and linked evidence &
Check grounding and evidence--argument consistency; fail closed on missing
or unavailable verification &
$v_a$ and uncertainty codes $\upsilon_a$ in $\mathcal{G}$ \\

A-QBAF reasoning &
Scored arguments for one subclaim &
Construct and propagate $Q_k$; identify high-weight clashes &
$Q_k$, $d_k$, and uncertainty codes \\

Decision / Report &
All $d_k$, dataset contract, and case trace &
Aggregate the final decision and confidence; render structured and readable reports &
$Y$, $\gamma$, $\mathcal{J}$, $\xi$, and $\mathcal{T}$ \\

Contestation &
Human actions, provenance links, and dependency graph &
Route revision to the earliest affected stage and recompute the affected scope &
$\mathcal{H}$, revised downstream state, and before/after trace \\

Reflection / Memory &
Finalized report and eligible post-prediction supervision &
Diagnose outcomes, verify lessons, stage candidates, and consolidate
independent observations &
Staged candidates and updated long-term memory in learning/bootstrap runs \\

\bottomrule
\end{tabularx}
\end{table}

\newpage
\subsection{Tool Configuration}\label{app:tools}
Table~\ref{tab:tool_versions} summarizes the software tools and runtime
configurations used by the SEMV media-processing, forensic, reverse-search, and model-serving components. Tool settings were held fixed across the reported evaluations, with optional components explicitly indicated.

\begin{table}[h]
\centering
\caption{Tool inventory and development versions.}
\label{tab:tool_versions}
\footnotesize
\setlength{\tabcolsep}{5pt}
\renewcommand{\arraystretch}{1.12}
\rowcolors{2}{white}{gray!5}
\begin{tabularx}{\textwidth}{>{\raggedright\arraybackslash}p{2.0cm}>{\raggedright\arraybackslash}p{5cm}>{\raggedright\arraybackslash}X}
\toprule
\textbf{Function} & \textbf{Tool/version} & \textbf{Configuration} \\
\midrule
Metadata & ExifTool 13.25; FFmpeg/FFprobe 7.1.1; Pillow 11.3.0 & embedded metadata and stream inspection \\
Frame extraction & PySceneDetect 0.6.6; FFmpeg 7.1.1 & scene detection, at most $K_f=8$ keyframes \\
OCR & PaddleOCR 3.3.2; PaddlePaddle 3.2.1 & PP-OCRv5 multilingual recognition with language/script-specific model selection\\
ASR & faster-whisper 1.1.1 & base model; language auto-detection \\
Forensics & Pillow 11.3.0; TruFor (optional) &
Error Level Analysis (ELA), noise/blur indicators, border or recapture cues, and editing-metadata checks; optional TruFor threshold $0.50$ \\
Local reverse search & ImageHash 4.3.2; open-clip-torch 2.32.0; FAISS 1.9.0 & pHash Hamming threshold 10; OpenCLIP ViT-B/32; FAISS inner-product threshold $0.84$ \\
Serving & vLLM 0.24.0; PyTorch 2.11.0 & BF16, batch 1, 32$K$-token context \\
\bottomrule
\end{tabularx}
\end{table}

\subsection{Configuration Constants}\label{app:constants}
Table~\ref{tab:constants} summarizes the fixed configuration constants used throughout SEMV. These settings govern media processing, evidence and memory ranking, argument scoring, decision thresholds, and memory verification. They were tuned on the COSMOS validation development subset and were kept unchanged across the reported COSMOS test, MV2026 validation/challenge, and CTR evaluations.

\begin{table}[htbp]
\label{tab:constants}
\centering
\caption{Configuration constants. Values were tuned on the COSMOS validation development subset.}
\label{tab:constants}
\footnotesize
\setlength{\tabcolsep}{4.5pt}
\renewcommand{\arraystretch}{1.08}
\rowcolors{2}{white}{gray!5}
\begin{tabularx}{\textwidth}{p{2.7cm}X}
\toprule
\textbf{Component} & \textbf{Value(s)} \\
\midrule
Media/retrieval &
$K_f=8$, $K_m=5$, top-10 evidence; frame-dedup pHash 4;
reverse-search pHash 10; CLIP 0.84; forensic 0.50 \\

Evidence ranking &
$+0.20$ claim-type, up to $+0.20$ lexical overlap,
$-0.15$ uncertainty \\

Diversity/specificity &
source cap 3; modality cap 2;
$p_a=\clip(|t_a|/280,0.25,1)$ \\

Clash/decision &
clash threshold $0.55$ with margin $\leq0.25$;
decision bands $0.30,0.45,0.55,0.70$ \\

Confidence &
verified 0.65; manipulated/false-context/mostly verified 0.55;
partially verified 0.45; insufficient evidence 0.35 \\

Promotion &
$(\theta_\tau^c,\theta_\tau^x,\theta_\tau^s)$:
episodic $(0.85,1,1)$; failure $(0.70,2,2)$;
semantic $(0.75,3,3)$ \\
\bottomrule
\end{tabularx}
\end{table}

\newpage
\subsection{Stability Evaluation}

Table~\ref{tab:run_stability} summarizes the three-run stability evaluation. Across the fixed COSMOS subset, run-to-run variation is small, with backbone accuracy standard deviations of 0.20--0.42 percentage points. Memory-policy accuracy variation is similarly limited, with standard deviations of 0.20--0.53 percentage points; the largest variation occurs for unverified append-all and the smallest for full verified consolidation.

On the 10-case MV2026 held-out validation set, Nemotron3 and InternVL3.5 retain the same number of correct decisions across all three runs, while Gemma4 and Qwen3.6 vary by at most one case. The sample standard deviation of Evidence-URL F1 ranges from 0.5 to 0.9 points across backbones. Overall, these results indicate limited sensitivity to run-level stochasticity under the evaluated settings. Given the small number of runs and MV2026 cases, however, they should be interpreted as descriptive stability evidence rather than population-level uncertainty estimates.

\begin{table}[h]
\centering
\caption{Three-run stability evaluation. For COSMOS, values are sample standard deviations of accuracy, in percentage points, across three runs on the same fixed stratified 500-case subset. Full-test results reported elsewhere remain single-run point estimates. For MV2026, correct-count ranges are reported because of the small held-out set ($N=10$); $s_{\mathrm{F1}}$ denotes the sample standard deviation of Evidence-URL F1 across the three runs, in F1 points on the 0--100 scale.}
\label{tab:run_stability}
\footnotesize
\setlength{\tabcolsep}{5pt}
\begin{tabularx}{\textwidth}{l l X}
\toprule
\textbf{Evaluation} & \textbf{Model} & \textbf{Observed Variation} \\
\midrule
\multirow{4}{*}{Backbones}
 & Nemotron3   & 0.31 pp \\
 & Gemma4      & 0.42 pp \\
 & InternVL3.5 & 0.31 pp \\
 & Qwen3.6     & 0.20 pp \\
\midrule
\multirow{5}{*}{Memory policy}
 & Off                           & 0.23 pp \\
 & Episodic only                 & 0.23 pp \\
 & Unverified append-all         & 0.53 pp \\
 & Staged, no conflict retention & 0.31 pp \\
 & Full verified consolidation   & 0.20 pp \\
\midrule
\multirow{4}{*}{MV2026}
 & Nemotron3   & 8--8/10 correct; $s_{\mathrm{F1}}=0.6$ \\
 & Gemma4      & 7--8/10 correct; $s_{\mathrm{F1}}=0.9$ \\
 & InternVL3.5 & 8--8/10 correct; $s_{\mathrm{F1}}=0.5$ \\
 & Qwen3.6     & 8--9/10 correct; $s_{\mathrm{F1}}=0.7$ \\
\bottomrule
\end{tabularx}
\end{table}

\newpage
\subsection{Sensitivity to Configuration Constants} \label{subsec:sensitivity}

Sec.~\ref{sec:method} introduces several decision thresholds and hand-set weights. We examine their local sensitivity in the fixed-stratified 500-image COSMOS validation analysis subset. 

All constants are frozen before the sensitivity sweep, and the subset is not used to revise them afterward. We divide the analysis according to how each factor is evaluated.

First, we perturb a factor only during decision-time inference while holding the retrieved evidence and consolidated-memory snapshot fixed (see Table~\ref{tab:sensitivity}a). Some factors, including the decision bands and argument-strength weights, can also affect later memory candidates through the resulting decisions, confidence values, and argument contributions. Their results therefore isolate the conditional decision-time component rather than the complete effect of reconstructing memory under the perturbed values. 

Second, we evaluate constants whose explicit role is to verify, admit, or promote experiences into long-term memory (see Table~\ref{tab:sensitivity}b). For every setting, we rebuild consolidated memory from the same fixed 1,200-case memory subset, freeze the resulting snapshot, and evaluate it on the same 500-case validation analysis subset.

\begin{table}[h] 
\centering 
\caption{Sensitivity to one factor group at a time on the fixed stratified $500$-image COSMOS validation subset with Qwen3.6. The reference configuration obtains $91.8\%$ accuracy. (A) isolates conditional decision-time effects using identical evidence and the reference-built memory snapshot. (B) reconstructs memory from the same pre-test memory-learning partition for each setting. $\Delta$ denotes the difference from the validation reference in percentage points.} 
\label{tab:sensitivity} 
\footnotesize
\setlength{\tabcolsep}{4.5pt} 
\renewcommand{\arraystretch}{1.08} 
\begin{tabularx}{\textwidth}{p{2.35cm}Xcc} 
\toprule 
\textbf{Group} & \textbf{Perturbation} & \textbf{Acc. (\%)} & \textbf{$\Delta$ pp} 
\\ 
\midrule 
\rowcolor{PalePurple!6}\multicolumn{4}{l}{ \textit{(a) Conditional decision-time sensitivity with fixed evidence and reference-built memory}} 
\\ 
\addlinespace[2pt] 
Decision bands & Shift all decision cut-points by $-0.05$ / $+0.05$ & $91.6$ / $92.0$ & $-0.2$ / $+0.2$ \\ Strength ablation & Remove reliability / relevance & $91.2$ / $91.0$ & $-0.6$ / $-0.8$ \\ Strength ablation & Remove source diversity / modality diversity & $91.6$ / $92.0$ & $-0.2$ / $+0.2$ \\ Strength ablation & Remove provenance / specificity & $91.2$ / $91.8$ & $-0.6$ / $0.0$ \\ Strength control & Replace the six hand-set weights with uniform $1/6$ weights & $91.4$ & $-0.4$ \\ \addlinespace[3pt] 
\midrule
\rowcolor{PalePurple!6}\multicolumn{4}{l}{\textit{(b) End-to-end memory-policy sensitivity with memory rebuilt for each setting}} \\ \addlinespace[2pt] Memory promotion & Shift all promotion-confidence thresholds by $-0.05$ / $+0.05$ & $91.6$ / $92.0$ & $-0.2$ / $+0.2$ \\ Memory verification & Set $\theta_v=0.50/0.60/0.70$ & $91.2$ / $91.8$ / $91.6$ & $-0.6$ / $0.0$ / $-0.2$ \\ 
\bottomrule 
\end{tabularx}
\end{table}

\begin{figure}[h]
\centering
\begin{minipage}[t]{0.30\linewidth}
    \vspace{0pt}
    \captionsetup{skip=0pt}
    \caption{\textit{Sensitivity to one factor group at a time on the COSMOS validation sensitivity subset using Qwen3.6.} Decision-time perturbations reuse fixed evidence and the reference-built memory snapshot, whereas memory-policy perturbations rebuild memory from the same pre-test memory-learning partition. Values report the change relative to the $91.8\%$ validation reference, followed by the resulting absolute accuracy in parentheses.}
    \label{fig:configuration_sensitivity}
\end{minipage}
\hfill
\begin{minipage}[t]{0.69\linewidth}
    \vspace{0pt}
    \centering
    \includegraphics[width=\linewidth]{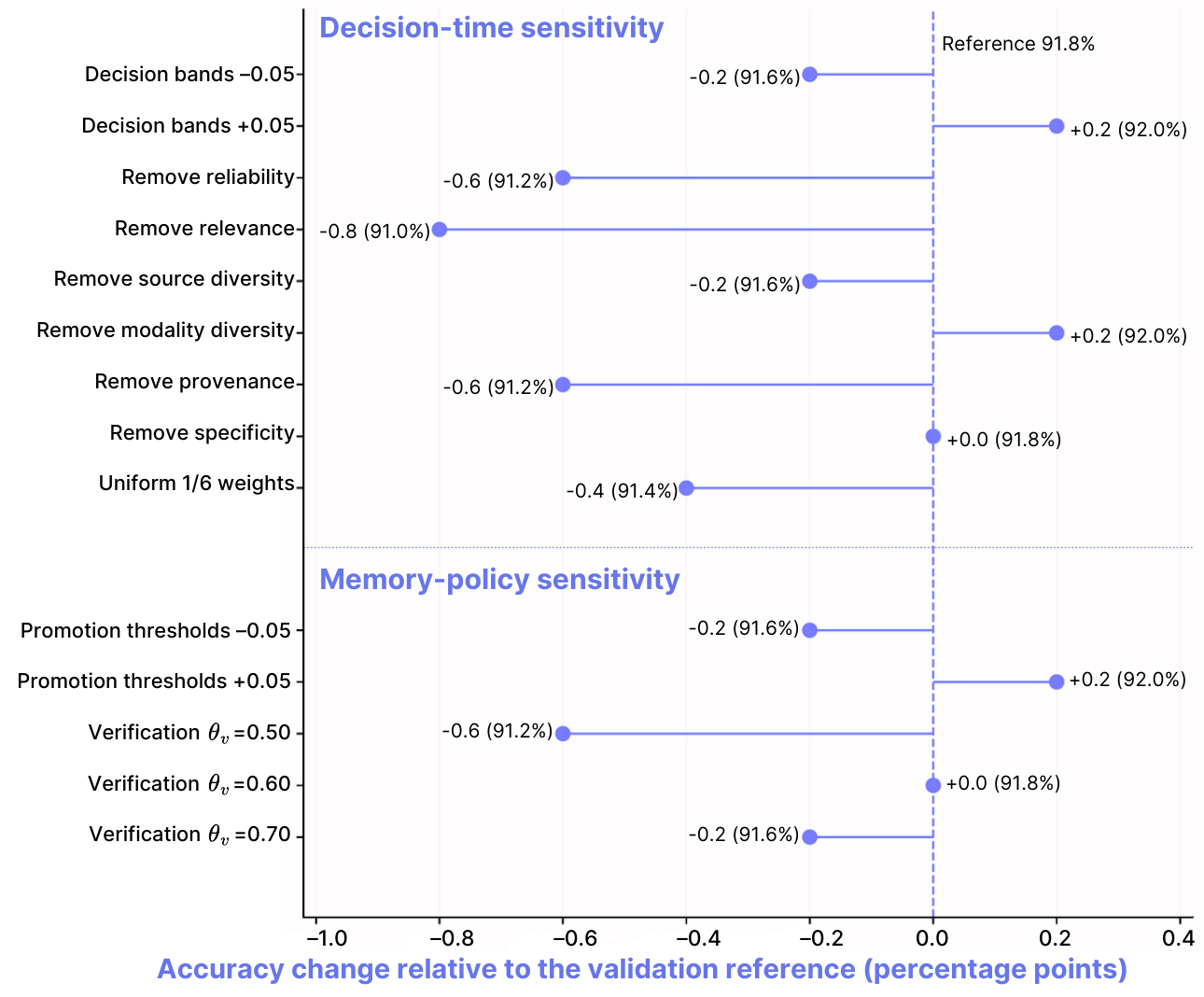}
\end{minipage}
\end{figure}

\paragraph{Results} As summarized in Fig.~\ref{fig:configuration_sensitivity}, 13 exploratory perturbations produce accuracies between $91.0\%$ and $92.0\%$, remaining within $0.8$ percentage points of the $91.8\%$ validation reference. The largest reduction occurs when relevance is removed, followed by reliability and provenance. In contrast, increasing the decision bands by $0.05$, removing modality diversity, and increasing the promotion thresholds each improve validation accuracy by $0.2$ percentage points. Removing specificity leaves accuracy unchanged at the reported precision, while uniform weighting reduces accuracy by $0.4$ percentage points. 

We treat this sweep as a descriptive development analysis and do not interpret the point-estimate differences as evidence that any individual component is statistically necessary. The reference configuration is selected on the separate development pool and frozen before the 500-case analysis subset is examined; sensitivity results are not used to retune it before final test evaluation. Because the analysis varies one factor group at a time, it does not characterize interactions among constants or robustness under larger configuration changes. Within the examined neighborhood, validation accuracy is locally stable.
\end{document}